%% file: main.tex
\documentclass[a4paper,USenglish,cleveref, autoref, thm-restate]{lipics-v2021}
\usepackage[utf8]{inputenc}

\usepackage{array}
\usepackage{cite}

\nolinenumbers
\usepackage{forest}
\usepackage{tikz}
\usepackage{booktabs}
\usepackage{subcaption}
\usetikzlibrary{patterns,decorations}
\usepackage{pgfplots}

\usepackage{graphicx}
\usepackage{amsmath,amsthm,amssymb,mathrsfs, thmtools}

\usepackage{mathtools}

\usetikzlibrary{positioning, arrows.meta, calc, backgrounds, shapes.geometric}

\usepackage{microtype}

\usepackage[noEnd,commentColor=blue,rightComments=false]{algpseudocodex}
\usepackage{algorithm}

\usepackage{bbm}
\usepackage{pythonhighlight}

\definecolor{codegreen}{rgb}{0,0.6,0}
\definecolor{codegray}{rgb}{0.5,0.5,0.5}
\definecolor{codepurple}{rgb}{0.58,0,0.82}
\definecolor{backcolour}{rgb}{0.95,0.95,0.92}

\lstdefinestyle{mystyle}{
    backgroundcolor=\color{backcolour},   
    commentstyle=\color{codegreen},
    keywordstyle=\color{magenta},
    numberstyle=\tiny\color{codegray},
    stringstyle=\color{codepurple},
    basicstyle=\ttfamily\footnotesize,
    breakatwhitespace=false,         
    breaklines=true,                 
    captionpos=b,                    
    keepspaces=true,                 
    numbers=left,                    
    numbersep=5pt,                  
    showspaces=false,                
    showstringspaces=false,
    showtabs=false,                  
    tabsize=2
}

\newcommand{\symdiff}{{\bigtriangleup}}

\usepackage{listings}

\definecolor{maroon}{cmyk}{0, 0.87, 0.68, 0.32}
\definecolor{halfgray}{gray}{0.55}
\definecolor{ipython_frame}{RGB}{207, 207, 207}
\definecolor{ipython_bg}{RGB}{247, 247, 247}
\definecolor{ipython_red}{RGB}{186, 33, 33}
\definecolor{ipython_green}{RGB}{0, 128, 0}
\definecolor{ipython_cyan}{RGB}{64, 128, 128}
\definecolor{ipython_purple}{RGB}{170, 34, 255}

\definecolor[named]{lipicsYellow}{rgb}{0.99,0.78,0.07}

\lstdefinelanguage{iPython}{
    morekeywords={access,and,break,class,continue,def,del,elif,else,except,exec,finally,for,from,global,if,import,in,is,lambda,not,or,pass,print,raise,return,try,while},%
    morekeywords=[2]{abs,all,any,basestring,bin,bool,bytearray,callable,chr,classmethod,cmp,compile,complex,delattr,dict,dir,divmod,enumerate,eval,execfile,file,filter,float,format,frozenset,getattr,globals,hasattr,hash,help,hex,id,input,int,isinstance,issubclass,iter,len,list,locals,long,map,max,memoryview,min,next,object,oct,open,ord,pow,property,range,raw_input,reduce,reload,repr,reversed,round,set,setattr,slice,sorted,staticmethod,str,sum,super,tuple,type,unichr,unicode,vars,xrange,zip,apply,buffer,coerce,intern},%
    sensitive=true,%
    morecomment=[l]\#,%
    morestring=[b]',%
    morestring=[b]",%
    morestring=[s]{'''}{'''},
    morestring=[s]{"""}{"""},
    morestring=[s]{r'}{'},
    morestring=[s]{r"}{"},%
    morestring=[s]{r'''}{'''},%
    morestring=[s]{r"""}{"""},%
    morestring=[s]{u'}{'},
    morestring=[s]{u"}{"},%
    morestring=[s]{u'''}{'''},%
    morestring=[s]{u"""}{"""},%
    literate=
    {á}{{\'a}}1 {é}{{\'e}}1 {í}{{\'i}}1 {ó}{{\'o}}1 {ú}{{\'u}}1
    {Á}{{\'A}}1 {É}{{\'E}}1 {Í}{{\'I}}1 {Ó}{{\'O}}1 {Ú}{{\'U}}1
    {à}{{\`a}}1 {è}{{\`e}}1 {ì}{{\`i}}1 {ò}{{\`o}}1 {ù}{{\`u}}1
    {À}{{\`A}}1 {È}{{\'E}}1 {Ì}{{\`I}}1 {Ò}{{\`O}}1 {Ù}{{\`U}}1
    {ä}{{\"a}}1 {ë}{{\"e}}1 {ï}{{\"i}}1 {ö}{{\"o}}1 {ü}{{\"u}}1
    {Ä}{{\"A}}1 {Ë}{{\"E}}1 {Ï}{{\"I}}1 {Ö}{{\"O}}1 {Ü}{{\"U}}1
    {â}{{\^a}}1 {ê}{{\^e}}1 {î}{{\^i}}1 {ô}{{\^o}}1 {û}{{\^u}}1
    {Â}{{\^A}}1 {Ê}{{\^E}}1 {Î}{{\^I}}1 {Ô}{{\^O}}1 {Û}{{\^U}}1
    {œ}{{\oe}}1 {Œ}{{\OE}}1 {æ}{{\ae}}1 {Æ}{{\AE}}1 {ß}{{\ss}}1
    {ç}{{\c c}}1 {Ç}{{\c C}}1 {ø}{{\o}}1 {å}{{\r a}}1 {Å}{{\r A}}1
    {€}{{\EUR}}1 {£}{{\pounds}}1
    {^}{{{\color{ipython_purple}\^{}}}}1
    {=}{{{\color{ipython_purple}=}}}1
    {+}{{{\color{ipython_purple}+}}}1
    {*}{{{\color{ipython_purple}$^\ast$}}}1
    {/}{{{\color{ipython_purple}/}}}1
    {+=}{{{+=}}}1
    {-=}{{{-=}}}1
    {*=}{{{$^\ast$=}}}1
    {/=}{{{/=}}}1,
    literate=
    *{-}{{{\color{ipython_purple}-}}}1
     {?}{{{\color{ipython_purple}?}}}1,
    identifierstyle=\color{black}\ttfamily,
    commentstyle=\color{ipython_cyan}\ttfamily,
    stringstyle=\color{ipython_red}\ttfamily,
    keepspaces=true,
    showspaces=false,
    showstringspaces=false,
    rulecolor=\color{ipython_frame},
    frame=single,
    frameround={t}{t}{t}{t},
    xleftmargin=.1\textwidth, xrightmargin=.1\textwidth,
    numbers=left,
    numberstyle=\tiny\color{halfgray},
    backgroundcolor=\color{ipython_bg},
    basicstyle=\scriptsize,
    keywordstyle=\color{ipython_green}\ttfamily,
}

\newcommand{\Cay}{\textsf{Cay}}

\DeclareMathOperator{\dist}{\textsf{dist}}
\DeclareMathOperator{\diam}{\textsf{diam}}

\renewcommand{\leq}{\leqslant}
\renewcommand{\geq}{\geqslant}
 \pgfplotsset{compat=1.18}

\title{Price of Locality in Permutation Mastermind:\texorpdfstring{\\}{ } Are TikTok influencers Chaotic Enough?} 
\titlerunning{Price of Locality in Permutation Mastermind}

\author{Bernardo Subercaseaux}{Carnegie Mellon University}{bersub@cmu.edu}{https://orcid.org/0000-0003-2295-1299}{}

\authorrunning{B. Subercaseaux} 

\Copyright{Bernardo Subercaseaux} 

\ccsdesc[500]{Mathematics of computing~Combinatorics}

\keywords{Permutation Mastermind, Locality, NP-hard} 

\category{} 

\relatedversion{} 

\funding{This work is supported by the National Science Foundation under grant DMS-2434625.}

\EventEditors{John Iacono}
\EventNoEds{1}
\EventLongTitle{13th International Conference on Fun with Algorithms (FUN 2026)}
\EventShortTitle{FUN 2026}
\EventAcronym{FUN}
\EventYear{2026}
\EventDate{May 18--22, 2026}
\EventLocation{Porquerolles, France}
\EventLogo{}
\SeriesVolume{366}
\ArticleNo{35}
\begin{document}

\maketitle 
\begin{abstract}
In the permutation Mastermind game, the goal is to uncover a secret permutation $\sigma^\star \colon [n] \to [n]$ by making a series of guesses $\pi_1, \ldots, \pi_T$ which must also be permutations of $[n]$, and receiving as feedback after guess $\pi_t$ the number of positions $i$ for which $\sigma^\star(i) = \pi_t(i)$.
While the existing literature on permutation Mastermind suggests strategies in which $\pi_t$ and $\pi_{t+1}$ might be widely different permutations, a resurgence in popularity of this game as a TikTok trend shows that humans (or at least TikTok influencers) use strategies in which consecutive guesses are very similar. For example, it is common to see players attempt one transposition at a time and slowly see their score increase.
Motivated by these observations, we study the theoretical impact of two forms of \emph{locality} in permutation Mastermind strategies: $\ell_k$-local strategies, in which any two consecutive guesses differ in at most $k$ positions, and the even more restrictive class of $w_k$-local strategies, in which consecutive guesses differ in a window of length at most $k$.
We show that, in broad terms, the optimal number of guesses for local strategies is quadratic, and thus much worse than the $O(n \lg n)$ guesses that suffice for non-local strategies. 
We also show NP-hardness of the satisfiability version for $\ell_3$-local strategies, whereas in the $\ell_2$-local variant the problem admits a randomized polynomial-time algorithm.
\end{abstract}

\section{Introduction}

The classical game of Mastermind, and its many variations, have received significant attention from computer scientists and mathematicians since the 1970s~\cite{Knuth1977TheCA,doerr2016playing,martinsson2024mastermind,jager2009number, el2020exact,afshani2019query,goodrich2009algorithmic, larcher2022solvingstaticpermutationmastermind}. However, it is only in the last couple of years that the \emph{permutation black-peg Mastermind} variant has reached internet virality through a TikTok challenge often named ``bottle match challenge'' or ``color match challenge''. The goal of the challenge is to guess a secret permutation of colored bottles, based on guesses that also consist of permutations of an equivalent set of colored bottles. To each guess, a different player that we will call \emph{codemaker}, responds with the number of positions for which the guess matches the secret permutation. In the TikTok trend, the secret permutation is hidden from the player under a table or a box, but visible to the viewer throughout the game (see~\Cref{fig:tiktok}). Perhaps the success of the trend can be explained by paraphrasing filmmaker Alfred Hitchcock: suspense is not about a bomb exploding on scene, but rather about the audience watching people talk about trivial matters while a bomb ticks underneath their table.  
 
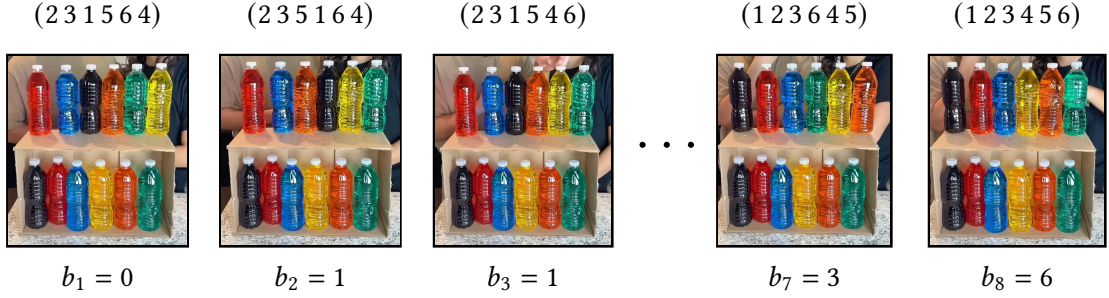
\begin{figure}[ht]
    \centering
    \begin{tikzpicture}
    \def\fwidth{1.9cm}
    \def\hsp{0.3}
            
        \node[draw, line width=2pt, inner sep=0pt] (img1) at (0, 0) {
        \includegraphics[height=\fwidth]{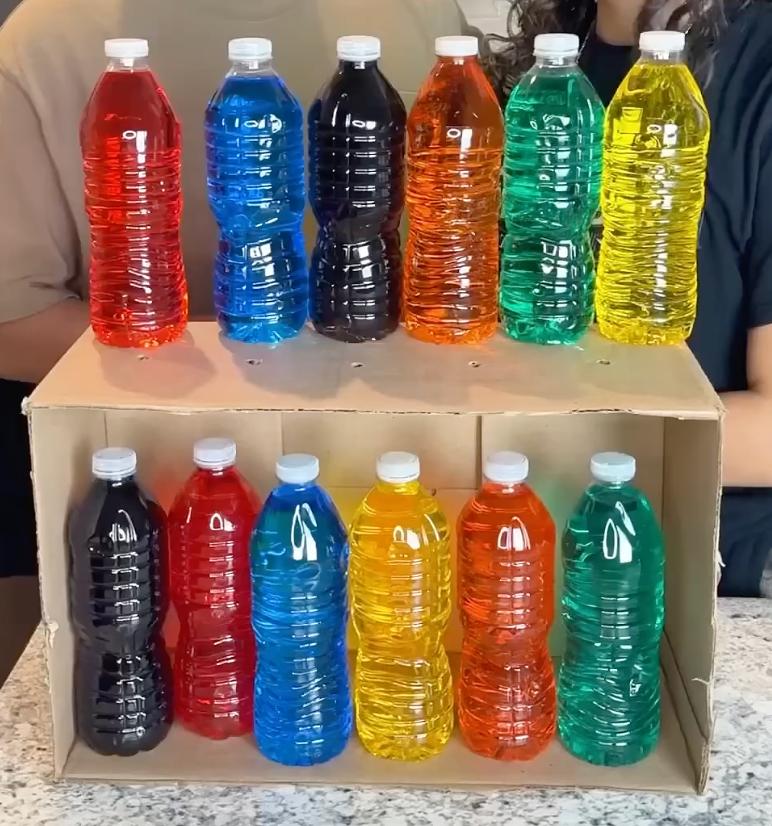}
        };
    \node at (img1.north)[above=0.2cm] {$(2 \, 3 \, 1 \, 5 \, 6\, 4)$};
     \node at (img1.south)[below=0.1cm] {$b_1 = 0$};
        
         \node[draw, line width=2pt, inner sep=0pt] (img2) at (2.5+\hsp, 0) {
        \includegraphics[height=\fwidth]{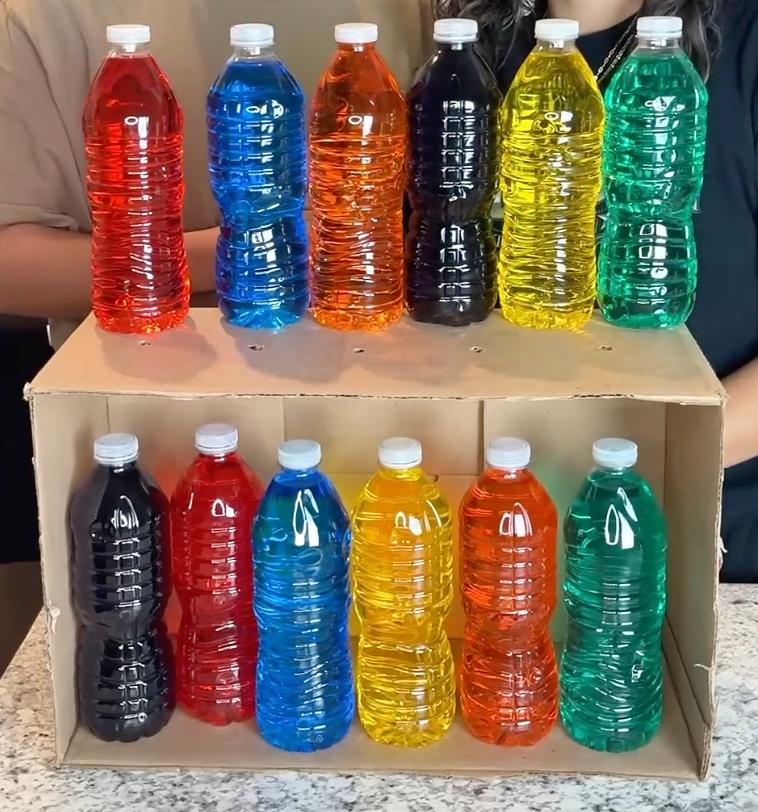}
        };

     \node at (img2.north)[above=0.2cm] {$(2 \, 3 \, 5 \, 1 \, 6\, 4)$};

     \node at (img2.south)[below=0.1cm] {$b_2 = 1$};
        
         \node[draw, line width=2pt, inner sep=0pt] (img3) at (5+2*\hsp, 0) {
        \includegraphics[height=\fwidth]{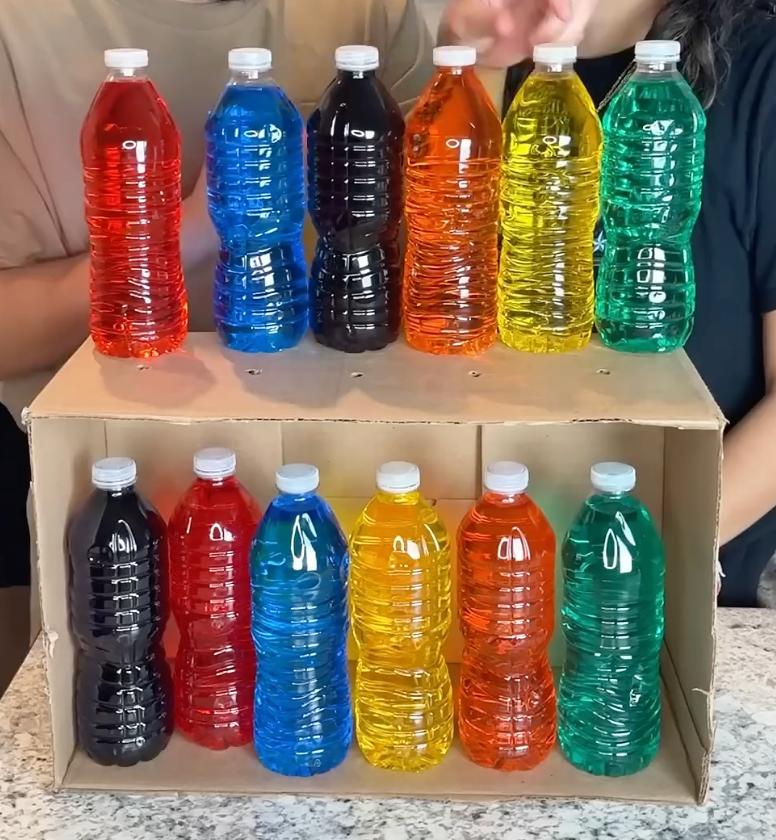}
        };

     \node at (img3.north)[above=0.2cm] {$(2 \, 3 \, 1 \, 5 \, 4\, 6)$};

      \node at (img3.south)[below=0.1cm] {$b_3 = 1$};

        \node (el) at (7.25+\hsp, 0) {\huge $\cdots$};

        \node[draw, line width=2pt, inner sep=0pt] (img4) at (9 +\hsp, 0) {
        \includegraphics[height=\fwidth]{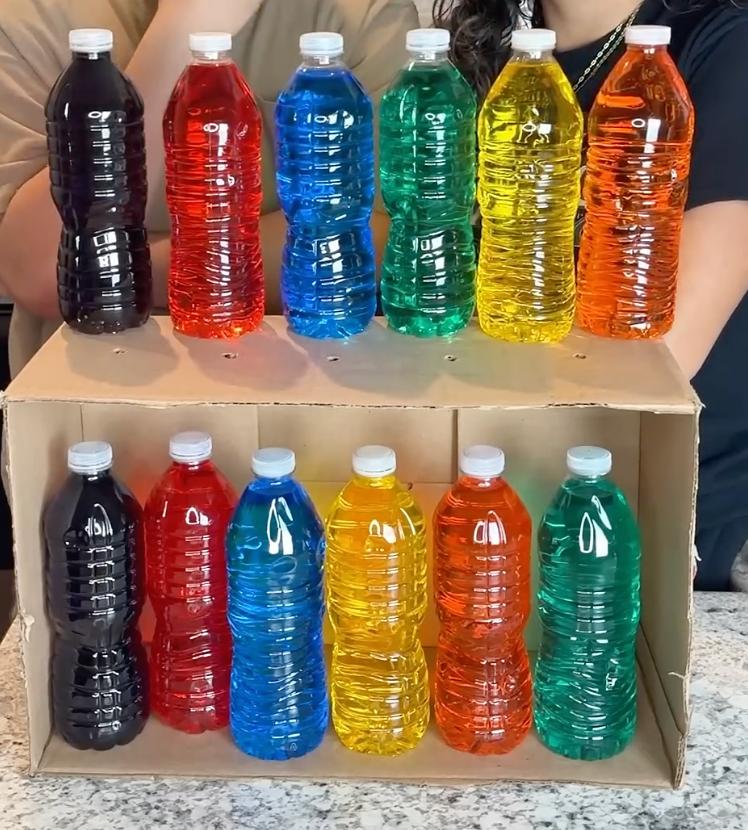}
        };

 \node at (img4.north)[above=0.2cm]
{$(1 \, 2 \, 3 \, 6 \, 4\, 5)$};

\node at (img4.south)[below=0.1cm] {$b_7 = 3$};

         \node[draw, line width=2pt, inner sep=0pt] (img5) at (11.5 +2*\hsp, 0) {
        \includegraphics[height=\fwidth]{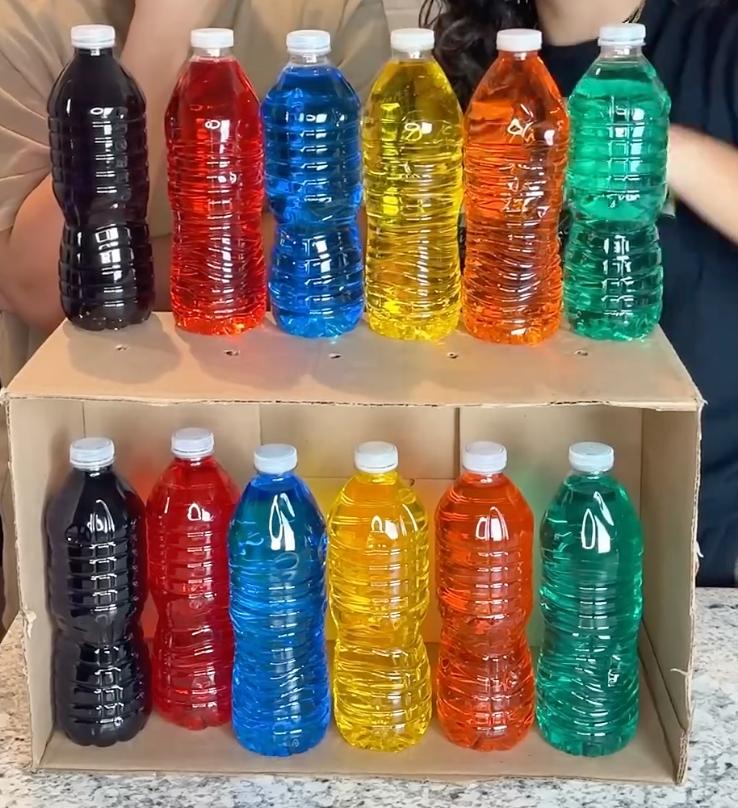}
        };

         \node at (img5.north)[above=0.2cm]
{$(1 \, 2 \, 3 \, 4 \, 5\, 6)$};

    \node at (img5.south)[below=0.1cm] {$b_8 = 6$};
    \end{tikzpicture}
    \caption{Illustration of a game of permutation Mastermind for $n = 6$ from TikTok~\cite{TikTokMake}. Guesses are labeled by treating the secret permutation as the identity $\sigma^\star := (1 \, 2 \, 3\, 4 \, 5 \, 6)$. }\label{fig:tiktok}
\end{figure}

After observing a variety of TikTok influencers play permutation Mastermind, a common aspect of their strategies came to my attention: \emph{they usually only permute a couple of bottles between guesses}. In contrast, the best known upper bound for permutation Mastermind, which takes $O(n \lg n)$ guesses for $n$ bottles, is based on guessing uniformly random permutations~\cite{larcher2022solvingstaticpermutationmastermind}, and thus almost all bottles are permuted between guesses on average.\footnote{Recall that a uniformly random permutation has a single fixed point in expectation.} My initial hypothesis was that these \emph{``local''} strategies, which permute only a small number of bottles between guesses, are natural for humans since it might be easier to keep track of the information throughout the game. This is reminiscent, for instance, of the analysis of Mastermind with constant-sized memory, by Doerr and Winzen~\cite{doerr2011playingmastermindconstantsizememory}. This paper aims to take the only reasonable course of action after TikTok has started to appropriate our well-studied game of Mastermind: to study a theoretical model of these local strategies and their implications both on the number of guesses required to win the game, and on its computational complexity.

\subsection{Preliminaries}
Before stating our main results, let us briefly formalize the problem at hand and present some of the best bounds known for the general case of permutation Mastermind.  

\newcommand{\Sn}{S_n}
\newcommand{\id}{\mathrm{id}}

\subparagraph*{Permutation Mastermind.}
Let $n\geq 2$ and $\Sn$ be the symmetric group, or in other words, the set of permutations $[n] \to [n]$ together with the composition operation.
A secret code is a permutation $\sigma^\star\in\Sn$ chosen by the \emph{codemaker}.
In round $t$, the \emph{codebreaker} plays a guess $\pi_t\in\Sn$ and receives the \emph{black-peg score}\footnote{The name comes from the original physical version of the game, in which \emph{black pegs} were used to indicate the number of correct positions guessed. }
\[
b(\pi_t,\sigma^\star)\;:=\; \bigl|\{i\in[n] : \pi_t(i)=\sigma^\star(i)\}\bigr|.
\]
An \emph{adaptive strategy} is a function that maps the \emph{transcript}
\(
\bigl((\pi_1,b_1),\dots,(\pi_{t-1},b_{t-1})\bigr)
\) 
to the next guess $\pi_t\in\Sn$.
A \emph{static} (non-adaptive) strategy is a fixed list of guesses $(\pi_1,\dots,\pi_T)$ chosen in advance.
In the adaptive setting, a guess $\pi_t$ wins the game for the codebreaker if $\pi_t = \sigma^\star$. In the static (i.e., non-adaptive) setting, after the codebreaker announces their fixed list of guesses $(\pi_1, \dots, \pi_T)$, the codemaker responds with the scores $b(\pi_1, \sigma^\star), \dots, b(\pi_T, \sigma^\star)$, after which the codebreaker gets to do one last guess to potentially win the game~\cite{Glazik1530289}.

\subparagraph*{Locality notions and Cayley-Mastermind.}

We consider two different forms of locality, defined next based on the notation $D(\pi, \sigma) := \{i \in [n] : \pi(i) \neq \sigma(i) \}$ for the set of indices in which two permutations differ. For any $k \geq 2$, we will say a strategy is $\ell_k$-local if for every two consecutive guesses $\pi_t$ and $\pi_{t+1}$ made by the strategy in any transcript, we have $|D(\pi_{t}, \pi_{t+1})| \leq k$. For a permutation $\pi \in S_n$, we call the set $D(\pi, \id_n) := \{ i \in [n]: \pi(i) \neq i\}$ the \emph{support} of $\pi$.

A strategy is $w_k$-local if for every two consecutive guesses $\pi_{t}$ and $\pi_{t+1}$ it holds that 
\[
\max(D(\pi_{t}, \pi_{t+1})) - \min(D(\pi_{t}, \pi_{t+1})) \leq k-1.
\]
Intuitively, this means there is some ``window'' $I := \{i, i+1, \ldots, i+k-1\}$ such that $\pi_{t+1}$ differs from $\pi_{t}$ only within that window (i.e., $D(\pi_{t}, \pi_{t+1}) \subseteq I$).

In the static case, we will assume that the last guess after receiving the scores is free of locality restrictions.

Our notions of locality are a particular case of playing Mastermind on a Cayley graph.
Given a set of generators $\Gamma$ of $S_n$, we can define the $\Gamma$-Permutation-Mastermind game in which the codemaker chooses a secret vertex of the Cayley graph $\Cay(S_n, \Gamma)$, then the codebreaker guesses an arbitrary starting vertex $v_0$, and for $t \ge 1$, the codebreaker must guess a vertex $v_t$ adjacent in $\Cay(S_n, \Gamma)$ to $v_{t-1}$. Equivalently, a vertex $v_t$ such that $v_{t-1}^{-1} v_{t} \in \Gamma$.
Naturally, we can define in turn a $\Gamma$-strategy for the standard permutation Mastermind as one that always plays according to the rules of the $\Gamma$-Permutation-Mastermind game. For example, for any $k \geq 2$, we define $L^{(k)}_n$ as the set of all permutations in $S_n$ whose support has size at most $k$. 
A strategy of permutation Mastermind is then said to be $\ell_k$-local if its consecutive guesses differ in at most $k$ positions, or equivalently, if it plays on the Cayley graph $\Cay(S_n, L^{(k)}_n)$. We purposely present this more general version of the definition to motivate a natural direction of future research: studying permutation Mastermind under other sets of generators.





\subparagraph*{Query complexities.}
Let $A(n)$ be the minimum $T$ such that there exists an adaptive strategy that determines
$\sigma^\star$ within $T$ guesses in the worst case.
Let $S(n)$ be the analogous minimum for static strategies.
For local versions, define:
\(
A_{\ell}(n,k),\;A_{w}(n,k),\;S_{\ell}(n,k),\;S_{w}(n,k)
\)
as the minimum worst-case number of guesses among adaptive/static strategies subject to the
corresponding locality constraint ($\ell_k$ or $w_k$).
%
The main previous results that contextualize our work are summarized in~\Cref{table:res}.

\begin{table}[h]
\caption{Best-known asymptotic query complexities for permutation Mastermind.}\label{table:res}
\centering
\small
\begin{tabular}{@{}lll@{}}
\toprule
Function &  Lower bound & Upper bound \\
\midrule
$A(n)$& 
$\Omega(n)$~\cite{KO1986449,ouali2016querycomplexityblackpegabmastermind} & $O(n\lg n)$~\cite{KO1986449,ouali2016querycomplexityblackpegabmastermind}  \\
$S(n) $ & 
$\Omega(n\lg n)$~\cite{Glazik1530289} & $O(n \lg n)$~\cite{larcher2022solvingstaticpermutationmastermind} \\
\bottomrule
\end{tabular}
\end{table}

\subsection{Our Contributions}

We prove the following results.

\begin{theorem}\label{thm:adaptive-ell}
    For the $\ell_k$-local setting, we have 
    \(    \frac{n^2 - 3n}{2k} \leq A_\ell(n, k) \leq \frac{n^2 \lg n}{k}(1+o(1)).
    \)
\end{theorem}

\begin{theorem}\label{thm:static}
    In the $\ell_k$-local static setting, we have 
    \[
    \frac{n^2 - (1+o(1))n^{3/2}}{k} \leq S_\ell(n, k)  \leq  28 n \lg n  \cdot \left\lceil\frac{n-1}{k-1}\right\rceil.
    \]
\end{theorem}

\begin{theorem}\label{thm:window}
    For the $w_k$-local setting, we have $S_w(n, 2) = \Theta(n^2)$. 
\end{theorem}

Our lower bounds are graph-theoretic in nature, and interestingly, they are based on giving the codebreaker even more advantage in order to simplify the analysis; we consider \emph{generous} codemakers that reveal exactly which of the permuted positions of some guess is correct. A similar setting was studied by~Li and Zhu~\cite{Li_2024}, who considered the more general case of Mastermind with ``Wordle'' feedback. Our upper bounds for~\Cref{thm:adaptive-ell} and~\Cref{thm:static} are based on simulating the best known upper bounds without locality restrictions. The upper bound of~\Cref{thm:window}, in turn, uses a similar strategy to that of~\cite{Li_2024}.

Finally, we prove a hardness result, showing that processing the feedback received is not easy even if each guess was a single transposition from the previous one. It was first proved by de Bondt~\cite{de2004np} that standard Mastermind, which includes the \emph{white-peg} score, is NP-hard. Independently, Stuckman and Zhang~\cite{stuckman2005mastermindnpcomplete} provided a different proof.
Then, Goodrich~\cite{goodrich2009algorithmic} extended the result to only black-peg score feedback. We take these results further by proving hardness in the permutation variant, and even on the $\ell_3$-local setting. In contrast, for the $\ell_2$-local setting, we show the problem can be solved in randomized polynomial time, leveraging an algorithm for perfect matchings with parity constraints from Geelen and Kapadia~\cite{Geelen2017}.

\begin{theorem}\label{thm:np-hardness}
    The \emph{satisfiability} problem for permutation Mastermind in the $\ell_3$-local setting is $\mathrm{NP}$-hard. In other words, given a transcript $((\pi_1, b_1), \dots, (\pi_t, b_t))$, where each pair of consecutive guesses differs in at most 3 positions, it is $\mathrm{NP}$-hard to decide whether there exists a secret $\sigma^\star$ such that $b(\pi_i, \sigma^\star) = b_i$ for every $1 \leq i \leq t$. In turn, for the $\ell_2$-local setting there is a randomized polynomial-time algorithm.
\end{theorem}

    







\input{adaptive}

\input{static}

\input{window-locality}

\input{complexity}

\bibliography{references}


\end{document}

%% file: adaptive.tex
\section{Adaptive Lower Bound on \texorpdfstring{$\ell_k$}{ell k}-local Strategies}\label{sec:adaptive}

This section is devoted to proving the lower bound of~\Cref{thm:adaptive-ell}. The main object we will handle in the proof is a bipartite graph representing the potential secrets  $\sigma^\star$ compatible with the information seen thus far. 
It will be convenient for this proof to think of $\sigma^\star$ as a string of length $n$ over an alphabet $\Sigma := \{s_1, \ldots, s_n\}$, so that \emph{positions} of $\sigma$ are just integers in $[n] := \{1, \ldots, n\}$ but the \emph{elements} of $\sigma$ have a different type, $\Sigma$. This way, we think of $\sigma^\star$ as a function from $[n]$ to $\Sigma$. 
To any $\sigma \colon [n] \to \Sigma$, we associate a bipartite graph $B_\sigma := ([n] \sqcup \Sigma, E)$,  where a pair $\{i, s_j\} \in E$ if and only if $\sigma(i) = s_j$. Note that these graphs $B_\sigma$ are perfect matchings.


Recall that $b(\pi, \sigma) = |\{i\in [n]: \pi(i) = \sigma(i) \} |$ is the black-peg score of a guess.
Let $b^{-1}(k, \sigma)$ be the set of all guesses $\pi \colon [n] \to \Sigma$ such that $b(\pi, \sigma) = k$.

Then, given a sequence of guesses $\Pi = (\pi_1, \ldots, \pi_t)$, to which the codemaker has answered $B := (b_1, \ldots, b_t)$, and some $\pi \colon [n] \to \Sigma$, we say $\pi$ is \emph{compatible} with $(\Pi, B)$ if $b(\pi_i, \pi) = b_i$ for $1 \leq i \leq t$. 

\begin{definition}[Graph of possible secrets]
    Given a sequence of guesses $\Pi = (\pi_1, \ldots, \pi_t)$, to which the codemaker has answered $B := (b_1, \ldots, b_t)$, we define the \emph{graph of possible secrets} \[
    P_t := \bigcup_{\sigma \text{ is compatible with } (\Pi, B)} B_{\sigma}.
    \]
    As a degenerate case, $P_0$ is the complete bipartite graph between $[n]$ and $\Sigma$.
\end{definition}

The following observation is immediate, but will be useful to translate our problem into a graph-theoretic one.
\begin{observation}
    Given a sequence of guesses $\Pi = (\pi_1, \ldots, \pi_t)$, to which the codemaker has answered $B := (b_1, \ldots, b_t)$, the codebreaker can be certain of the secret permutation $\sigma$ if and only if $P_t$ has a unique perfect matching.
\end{observation}

We now prove that, for a graph $P_t$ to have a unique perfect matching, it needs to be missing a significant number of edges. We will make use of the following folklore observation (illustrated in~\Cref{fig:unique-pm}): 

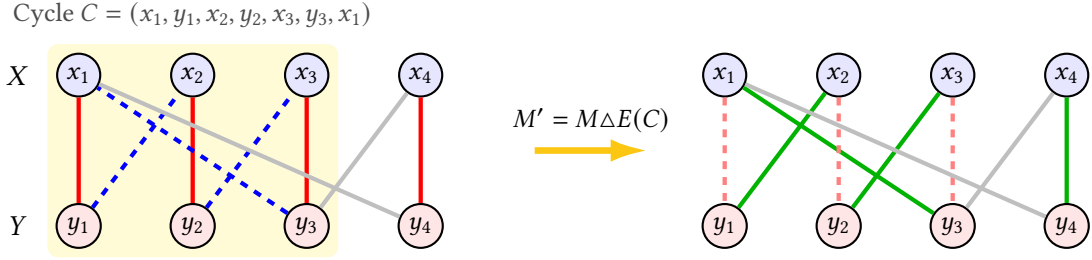
\begin{figure}
    \centering
    \input{figures/matching}
    \caption{Illustration of the proof for~\Cref{obs:unique}. Edges of the original matching $M$ are colored red, and edges in $E(C) \setminus M$ are dashed blue. Gray edges are not part of the matchings, and the resulting matching $M'$ is colored green.}
    \label{fig:unique-pm}
\end{figure}

\begin{observation}\label{obs:unique}
    If a finite bipartite graph $B$ has minimum degree at least $2$, it cannot have a unique perfect matching.
\end{observation}

\begin{proof}
If $B$ has no perfect matchings, the statement holds trivially. Thus, let $M$ be a perfect matching of $B$. For each vertex $u$, let $M(u)$ denote its neighbor in the matching $M$, and $D(u)$ denote one neighbor of $u$ via an edge not in $M$, chosen arbitrarily, which must exist since $u$ has degree at least $2$.
Now let $u_0$ be any vertex, and construct a walk $u_0, u_1, \ldots$ by
\[
u_i := \begin{cases}
    M(u_{i-1}) & \text{if } i \text{ is odd}\\
    D(u_{i-1}) & \text{otherwise}.
\end{cases}
\]
%
Since $B$ is finite, some vertex appears twice in this walk. Let $u_j=u_k$ with $j<k$, chosen so that all intermediate vertices are distinct. Then the subwalk $u_j,u_{j+1},\dots,u_k$ forms a simple alternating cycle $C$, which necessarily has even length. Consider the matching $M' = M \symdiff E(C)$. For vertices not in $C$, $M'$ agrees with $M$. Each vertex of $C$ is incident to exactly one edge of $M$ and one edge of $E(C) \setminus M$, so it is matched by exactly one edge in $M'$. Thus $M'$ is a perfect matching. Since $C$ contains edges not in $M$, we have $M' \neq M$, contradicting the uniqueness of $M$.
\end{proof}

We now use~\Cref{obs:unique} to show a much stronger condition implied by having a unique perfect matching.

\begin{lemma}\label{lemma:upm}
    Let $B = (X \sqcup Y, E)$ be a balanced bipartite graph with $n := |X| = |Y|$. Then, if $B$ has a unique perfect matching, $|E| \leq \binom{n+1}{2}$.
\end{lemma}
\begin{proof}
    Let $f(n)$ be the maximum number of edges of a balanced bipartite graph with $2n$ vertices that has a unique perfect matching. We trivially have $f(1) = 1$. For $n \geq 2$, we claim that $f(n) \leq f(n-1) + n$. Before proving the claim, let us see that it is enough to conclude, since then, using induction we will have
    \[
    f(n) \leq \binom{n}{2} + n = \frac{n(n-1)}{2} + n = \binom{n+1}{2}.
    \]
    Now, to prove the claim, we consider an arbitrary balanced bipartite graph $B_{n}$ on $2n$ vertices with a unique perfect matching and such that $|E(B_n)| = f(n)$. Let $\delta$ be the minimum degree in $B_n$, and note that $\delta = 1$: since $B_n$ has a perfect matching we have $\delta \geq 1$, and by~\Cref{obs:unique} we have $\delta < 2$. Let $u \in V(B_n)$ be a vertex of degree $1$, and $v$ its only neighbor. It must be the case that $u$ and $v$ belong to opposite sides of the bipartition, and thus if we consider the graph resulting from $B_n$ by removing vertices $u$ and $v$, we get a balanced bipartite graph $B'$ on $2(n-1)$ vertices. Moreover, $\{u,v\}$ must belong to the (unique) perfect matching of $B_n$, and therefore any perfect
matching of $B'$ would extend (by adding $\{u,v\}$) to a perfect matching of $B_n$, and thus $B'$ must also have a unique perfect matching. Therefore, $|E(B')| \leq f(n-1)$. On the other hand,
\[
|E(B')| = |E(B_n)| - \deg_{B_n}(v) - (\deg_{B_n}(u) - 1)  = |E(B_n)| - \deg_{B_n}(v) \geq |E(B_n)| - n.
\]

We thus have
    \[
    f(n) = |E(B_n)| \leq |E(B')| + n \leq f(n-1) + n,
    \]
concluding the claim.
\end{proof}

The tightness of~\Cref{lemma:upm} is witnessed by \emph{half graphs} $H_n$ (see~\Cref{fig:half-graph}), which include the edge $\{i, s_j\}$ if and only if $i \leq j$. The uniqueness of a perfect matching in $H_n$ can be easily seen by induction, noting that vertex $n$ has degree one, and upon removing $n$ together with its neighbor $s_n$, we are left exactly with $H_{n-1}$, for which the inductive hypothesis applies.

\begin{figure}
    \centering
    \input{figures/half_graph}
    \caption{The half graph $H_5$, with $\binom{6}{2} = 15$ edges. Its only perfect matching is highlighted in orange. }
    \label{fig:half-graph}
\end{figure}
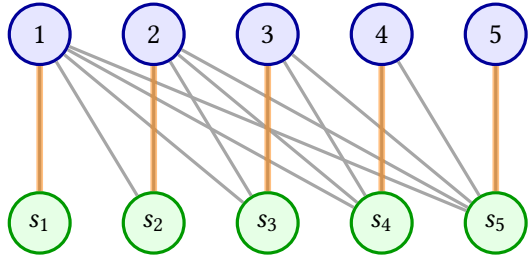

\subsection{Generous and Supergenerous Codemakers}

Perhaps paradoxically, I found that proving lower bounds on permutation Mastermind with locality becomes easier if we give more power to the codebreaker, but in a way that makes the game simpler.

Namely, while a normal codemaker will give $|B_{\pi} \cap B_{\sigma^\star}|$ as feedback (i.e., the number of correctly guessed positions), we will consider a \emph{generous codemaker} that gives the set $B_{\pi} \cap B_{\sigma}$ as feedback. 
Note that, any strategy against a generous codemaker can be used against a standard codemaker without any loss, since from $B_{\pi} \cap B_{\sigma}$ the codebreaker can compute $|B_{\pi} \cap B_{\sigma^\star}|$.

In the original form of the game, each guess $\pi_t$ can be described as well by $E_t := B_{\pi_t} \setminus B_{\pi_{t-1}}$, the set of newly guessed edges. Except for the first guess, which tests $n$ edges, each guess $\pi_t$ tests at most $k$ edges in the $\ell_k$ setting. We will define a \emph{supergenerous} codemaker as one that allows the codebreaker to guess individual edges one by one, giving generous feedback to each (i.e., revealing whether it is correct or not), but only accounting 1 guess per each $k$ individual-edge guesses. Naturally, if we prove a lower bound of $t$ individual-edge guesses against a supergenerous codemaker, then $(t-n)/k$ guesses are necessary against a standard codemaker.

\subsection{Finishing the proof of Theorem 1's lower bound}

\begin{lemma}
    Let $k \geq 2$. Then, any deterministic adaptive strategy in the $\ell_k$-local model requires $\frac{n^2-3n}{2k}$ guesses.
\end{lemma}
\begin{proof}
We consider an adversarial \emph{supergenerous} codemaker. The adversary will keep a graph $G_t$, with $G_0 := K_{n, n}$, and upon receiving a guess $e_t$ for $t \geq 1$, it will answer \emph{yes} if and only if $e_t$ belongs to all perfect matchings of $G_{t-1}$. If it answers \emph{no}, the adversary sets $G_t \gets G_{t-1} \setminus e_t$, and otherwise $G_t \gets G_{t-1}$. Let $\mathcal{S}_t$ be the set of all perfect matchings compatible with the answers given until time $t$ (including to guess $e_t$), and $\mathcal{S}_0 = \textsf{PM}(K_{n,n})$ the set of all perfect matchings of $K_{n,n}$. Then, by construction, we have that the adversary says \emph{yes} to an edge $e_t$ if and only if $e_t \in M$ for every $M \in \mathcal{S}_{t-1}$.

Let $\textsf{PM}(G_t)$ be the set of all perfect matchings in $G_t$. We will prove by induction that $\mathcal{S}_{t} = \textsf{PM}(G_t)$ for every $t \geq 0$. For $t = 0$ it holds trivially since both sets include all possible matchings. For $t \geq 1$, we consider two cases. First, if the adversary said \emph{no} to guess $e_t$,  we have
\[
\textsf{PM}(G_t) = \textsf{PM}(G_{t-1}) \setminus \{ M  \in \textsf{PM}(G_{t-1}) : e_t  \in M \},
\]
and similarly,
\(
\mathcal{S}_{t} = \mathcal{S}_{t-1} \setminus \{  M \in \mathcal{S}_{t-1} :  e_t \in M\}.
\)
But by the inductive hypothesis we then have
\[
\mathcal{S}_{t} = \textsf{PM}(G_{t-1}) \setminus \{ M \in \textsf{PM}(G_{t-1}) : e_t  \in M \} = \textsf{PM}(G_t).
\]
On the other hand, if the adversary answered \emph{yes}, we have by definition \begin{align*}
    \mathcal{S}_t &=  \mathcal{S}_{t-1} \setminus \{ M : M \in \mathcal{S}_{t-1} \text{ and } e_t \not\in M\} \\
    &= \mathcal{S}_{t-1} \setminus \varnothing  = \mathcal{S}_{t-1}\tag{Since the answer to $e_t$ was \emph{yes}}\\
    &= \textsf{PM}(G_{t-1}) \tag{inductive hypothesis} \\
    &= \textsf{PM}(G_t) \tag{Since $G_t = G_{t-1}$ as the answer to $e_t$ was \emph{yes}}.
\end{align*}

We have thus proved that $\mathcal{S}_t = \textsf{PM}(G_t)$. Now, observe that by definition the codebreaker can only win when $|\mathcal{S}_t| = 1$. But this implies $|\textsf{PM}(G_t)| = 1$, and by~\Cref{lemma:upm}, this implies $|E(G_t)| \leq \binom{n+1}{2}$. But $|E(G_t)| \geq |E(G_0)| - t$, from where we get
\[
t \geq |E(G_0)| - |E(G_t)| = n^2 -  |E(G_t)| \geq n^2 - \binom{n+1}{2} = \binom{n}{2}.
\]

Therefore, the lower bound we get against a standard codemaker is at least
\[
\frac{t-n}{k} \geq \frac{\binom{n}{2} -n}{k} = \frac{n^2 - 3n}{2k}. \qedhere
\]
\end{proof}

%% file: figures/matching.tex
\begin{tikzpicture}[
    thick,
    vertex/.style={circle, draw, fill=white, inner sep=2pt, minimum size=12pt, font=\small},
    edge_m/.style={ultra thick, red},          
    edge_c_not_m/.style={ultra thick, blue, dashed}, 
    edge_other/.style={ultra thick, gray!50},            
    edge_m_prime/.style={ultra thick, green!70!black} 
]


\begin{scope}[xshift=0cm, yshift=0cm]
    
    \foreach \i in {1,...,4} {
        \node[vertex, fill=blue!10] (u\i) at (\i*1.5 - 1.5, 1) {$x_{\i}$};
        \node[vertex, fill=red!10] (v\i) at (\i*1.5 - 1.5, -1) {$y_{\i}$};
    }
    \node at (-0.8, 1) {$X$};
    \node at (-0.8, -1) {$Y$};

    \draw[edge_m] (u1) -- (v1);
    \draw[edge_m] (u2) -- (v2);
    \draw[edge_m] (u3) -- (v3);
    \draw[edge_m] (u4) -- (v4);

    \draw[edge_c_not_m] (v1) -- (u2);
    \draw[edge_c_not_m] (v2) -- (u3);
    \draw[edge_c_not_m] (v3) -- (u1);
    
    \begin{pgfonlayer}{background}
        \fill[yellow!20, rounded corners] ($(u1.north west)+(-0.2,0.2)$) rectangle ($(v3.south east)+(0.2,-0.2)$);
    \end{pgfonlayer}
    \node[font=\small, gray!60!black, anchor=south] at (1.5,1.5) {Cycle $C = (x_1, y_1, x_2, y_2, x_3, y_3, x_1)$};

    \draw[edge_other] (u4) -- (v3);
    \draw[edge_other] (v4) -- (u1);


\end{scope}

\draw[-{Latex[length=5mm, width=3mm]}, line width=2.5pt, lipicsYellow] (5.5, 0) -- (7.0, 0);
\node[font=\small, align=center, above] at (6.25, 0.1) {$M' = M \symdiff E(C)$};

\begin{scope}[xshift=8.0cm, yshift=0cm]

    \foreach \i in {1,...,4} {
        \node[vertex, fill=blue!10] (ru\i) at (\i*1.5 - 1.5, 1) {$x_{\i}$};
        \node[vertex, fill=red!10] (rv\i) at (\i*1.5 - 1.5, -1) {$y_{\i}$};
    }

    
    \draw[edge_m_prime] (rv1) -- (ru2);
    \draw[edge_m_prime] (rv2) -- (ru3);
    \draw[edge_m_prime] (rv3) -- (ru1);
    
    \draw[edge_m_prime] (ru4) -- (rv4);

    \draw[edge_other, dashed, red!50] (ru1) -- (rv1);
    \draw[edge_other, dashed, red!50] (ru2) -- (rv2);
    \draw[edge_other, dashed, red!50] (ru3) -- (rv3);

    \draw[edge_other] (ru4) -- (rv3);
    \draw[edge_other] (rv4) -- (ru1);

    
\end{scope}

\end{tikzpicture}

%% file: figures/half_graph.tex
\def\n{5}
\begin{tikzpicture}[
    scale=0.8,
    xscale=1.5, 
    yscale=2.5, 
    u_node/.style={circle, draw=blue!60!black, very thick, fill=blue!10, minimum size=7mm, inner sep=1pt},
    v_node/.style={circle, draw=green!60!black, very thick, fill=green!10, minimum size=7mm, inner sep=1pt},
    myedge/.style={very thick, draw=gray!70}
]

    \foreach \i in {1,...,\n} {

        \node[u_node] (u\i) at (\i, 1) {${\i}$};

        \node[v_node] (v\i) at (\i, 0) {$s_{\i}$};
    }

    \foreach \i in {1,...,\n} {
        \foreach \j in {\i,...,\n} {
            \draw[myedge] (u\i) -- (v\j);
        }
    }

    \foreach \i in {1,..., \n} {
         \draw[orange, line width=2.5pt, opacity=0.5] (u\i) -- (v\i);
    }


\end{tikzpicture}

%% file: static.tex
\section{Lower Bound for the Static Variant}

This section is devoted to the static version of permutation Mastermind with locality considerations. Recall that in this variant, the codebreaker must reveal a fixed sequence of guesses  $\pi_1, \dots, \pi_T$ upfront, subject to the locality constraints, after which they will receive the black-peg feedback $b_1, \dots, b_T$, and then finally make their final guess that must equal the secret permutation $\sigma^\star$ in order to win. The final guess is not subject to the locality restrictions.

The main ingredient of our lower bound is a classic result on extremal graph theory (see~\cite{zhaoGraphTheoryAdditive2023} for a modern presentation).

\begin{theorem}[\cite{Kovari1954, zbMATH03232670}]\label{thm:k22}
    The maximum number of edges an $n$-vertex graph can have without containing a $K_{2,2}$ is $(\frac{1}{2} + o(1))n^{3/2}$. Moreover, the maximum number of edges on a $K_{2,2}$-free bipartite graph with $n$ vertices on each side is $(1 + o(1))n^{3/2}$.
\end{theorem}

For this proof we will consider a simpler graph than in~\Cref{sec:adaptive}. Let $U_t := ([n] \times \Sigma, E_t)$ be the bipartite graph where an edge $\{i, s_j\}$ is present if and only if no guess $\pi_r$ with $r \leq t$ has $\pi_r(i) = s_j$. Intuitively, $U_t$ corresponds to the \emph{untested} matching edges up to time $t$.
Now, we will prove the relevance of $K_{2,2}$-subgraphs. For this, let us say the static part of a strategy, $\pi_1, \ldots, \pi_T$, is \emph{sufficient}, if for every secret $\sigma^\star$, the only permutation compatible with the transcript 
\(
((\pi_1, b(\pi_1, \sigma^\star)), \dots, (\pi_T, b(\pi_T, \sigma^\star)))
\)
is exactly $\sigma^\star$. Naturally, the static part of any correct static strategy must be sufficient.

\begin{figure}
    \centering
    \input{figures/k22}
    \caption{Illustration of the obstacle introduced by a $K_{2,2}$ for uniquely determining $\sigma^\star$. On the left, tested edges are colored light gray, while untested edges are colored black, except for those forming a $K_{2,2}$ which are highlighted.}
    \label{fig:k22}
\end{figure}
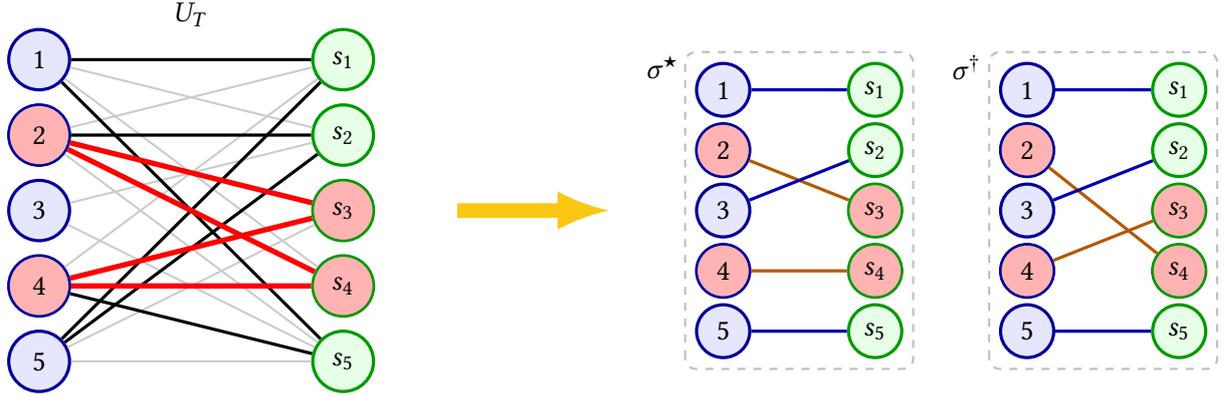

\begin{lemma}
    If $\pi_1, \ldots, \pi_T$ is the static part of a correct static strategy, then $U_T$ is $K_{2,2}$-free.
\end{lemma}
\begin{proof}
    Suppose for a contradiction that $U_T$ contains a $K_{2,2}$ consisting of the four edges
    \[
    \{i, s_x\}, \{i, s_y\}, \{j, s_x\}, \{j, s_y\}.
    \] An illustration with a concrete example is depicted in~\Cref{fig:k22}.
    Now, let $\sigma^\star$ be an arbitrary permutation such that $\sigma^{\star}(i) = s_x$ and $\sigma^{\star}(j) = s_y$. From $\sigma^\star$ we define now the permutation $\sigma^\dagger$ by
    $
    \sigma^\dagger(k) = \begin{cases}
        s_y & \text{if } k= i\\
        s_x & \text{if } k = j\\
        \sigma^\star(k) & \text{otherwise}.
    \end{cases}
    $
    In other words, $\sigma^\dagger$ is the matching obtained from  $\sigma^\star$ via the alternating $4$-cycle corresponding to the $K_{2,2}$. Now, note that for any  guess $\pi_t$ we have $b(\pi_t, \sigma^\star) = b(\pi_t, \sigma^\dagger)$. Indeed, let $k \in [n]$ and we prove by cases that $\sigma^\star(k) = \pi_t(k)$ if and only if $\sigma^{\dagger}(k) = \pi_t(k)$:
    \begin{itemize}
        \item If $k \not\in \{i, j\}$, then $\sigma^\star(k) = \sigma^\dagger(k)$, and thus the claim holds.
        \item If $k \in \{i, j\}$, then by construction  we have $\sigma^{\star}(k) \in \{s_x, s_y\}$ as well as $\sigma^{\dagger}(k) \in \{s_x, s_y\}$. But by definition of $U_T$, we have $\pi_t(k) \not\in \{s_x, s_y\}$. Therefore neither $\sigma^\star(k) = \pi_t(k)$ nor $\sigma^\dagger(k) = \pi_t(k)$.
    \end{itemize}
 We thus have that $b(\pi_t, \sigma^\star) = b(\pi_t, \sigma^\dagger)$ for any guess $\pi_t$, which implies both $\sigma^\star$ and $\sigma^\dagger$ are compatible with transcript $((\pi_1, b(\pi_1, \sigma^\star)), \dots, (\pi_T, b(\pi_T, \sigma^\star)))$, and thus $\pi_1,\dots, \pi_T$ is not sufficient.
\end{proof}


Our lower bound is now a direct consequence of the previous lemma and the quantitative bounds on $K_{2,2}$-free graphs.

\begin{theorem}
For any $k\geq 2,$ we have
     \(
     S_\ell(n, k) \geq \frac{n^2- (1+o(1))n^{3/2}}{k}.
     \)
\end{theorem}
\begin{proof}
    Consider a correct strategy whose static part is $\pi_1, \dots, \pi_T$.
    Note that before the first guess we have $|E(U_0)| = n^2$. The first guess tests $n$ edges, so $|E(U_1)| = n^2 - n$, and after this special first guess the $\ell_k$-locality condition implies for $2 \leq t \leq T$ that
    \[
    |E(U_{t})| \geq |E(U_{t-1})| - k.
    \]
    Therefore, 
    $|E(U_T)| \geq |E(U_1)| - k T = n^2 -n - kT$.
    Now, since the static strategy is correct, we must have that $U_T$ is $K_{2,2}$-free, and by~\Cref{thm:k22} it must hold that 
    \(
    |E(U_T)| \leq (1 + o(1))n^{3/2}.
    \)
    Combining our inequalities, we have
    \(
    (1 + o(1))n^{3/2} \geq n^2 -n - kT
    \), from where $T \geq \frac{n^2 - (1 + o(1))n^{3/2}}{k}$, thus concluding the proof.
\end{proof}

\section{Diameter-Based Upper Bounds}\label{sec:diameter-upper}

In this section we prove the upper bounds for~\Cref{thm:adaptive-ell} and~\Cref{thm:static} by a rather direct simulation adaptation of the $O(n \lg n)$ upper bounds of~\cite{larcher2022solvingstaticpermutationmastermind,ouali2016querycomplexityblackpegabmastermind}. In a nutshell, it suffices to show that we can pay $O(n/k)$ to simulate regular guesses in the $\ell_k$-local model. This idea generalizes to the following lemma.

\begin{lemma}\label{lemma:diam-prod}
    Let $\Gamma$ be a set of generators for $S_n$. Let  $S_\Gamma(n)$  denote the minimum number of guesses for a correct static strategy for the $\Gamma$-Permutation-Mastermind game, and $A_\Gamma(n)$ denote the minimum number of guesses in the worst case of an adaptive strategy. Then,
\[
    S_\Gamma(n) \leq S(n) \cdot \mathrm{diam}(\Cay(S_n, \Gamma)) \quad \text{and} \quad  A_\Gamma(n) \leq A(n) \cdot \mathrm{diam}(\Cay(S_n, \Gamma))
\]
\end{lemma}
\begin{proof}
    We write the proof in terms of the static setting, but the adaptive one is identical. Let $\pi_1, \dots, \pi_{S(n)}$ be a correct static strategy for the standard permutation Mastermind. Then, for each $1 \leq i \leq S(n)$, write $(\pi_{i-1})^{-1}\pi_{i}$ as a word $\omega_i := \gamma^{(i)}_1 \dots \gamma^{(i)}_{r_i}$ over $\Gamma$ that minimizes $r_i$. Then, by definition of the diameter $r_i \leq \text{diam}(\Cay(S_n, \Gamma))$. Now, construct the guess sequence
    \[
    \pi_\Gamma := \pi_1 \gamma^{(2)}_1 \dots \gamma^{(2)}_{r_2} \gamma^{(3)}_1 \dots \gamma^{(3)}_{r_3} \dots \gamma^{(S(n))}_1 \dots \gamma^{(S(n))}_{r_{S(n)}},
    \]
    which is $\Gamma$-Permutation-Mastermind strategy, and has size at most
    \[
   1 + \sum_{i=2}^{S(n)} r_i \leq 1 + (S(n)-1) \text{diam}(\Cay(S_n, \Gamma)) \leq  S(n) \cdot \text{diam}(\Cay(S_n, \Gamma)).\qedhere
    \]
    
\end{proof}

We now prove a tight bound on the diameter of the Cayley graph corresponding to the $\ell_k$-local setting.

\begin{lemma}\label{lem:routing}
Fix $k \geq 2$. Recall that $L^{(k)}_{n}$ is the set of permutations
with support on at most $k$ elements. Then, 
\[
    \mathrm{diam}(\Cay(S_n, L^{(k)}_{n})) \leq \left\lceil \frac{n-1}{k-1}\right\rceil.
\]

\end{lemma}

\begin{proof}
Since $\Cay(S_n, L^{(k)}_{n})$ is a Cayley graph and thus vertex transitive, it suffices to upper bound the distance between the identity permutation $\textsf{id}_n$, and an arbitrary permutation $\pi$.

Let $M := D(\pi,\mathrm{id}_n)=\{i\in[n]:\pi(i)\neq i\}$ be the support of $\pi$, and let
$m:=|M|$. 
Letting $\dist := \dist_{\mathrm{Cay}(S_n, L^{(k)}_n)}$, we
will prove by induction on $m$ that the stronger statement
\(
\dist(\mathrm{id},\pi) \leq \left\lceil \frac{m-1}{k-1}\right\rceil
\)
holds for any $\pi$.

If $m=0$, then $\pi=\mathrm{id}_n$ and the bound holds. Otherwise $m\geq 2$, since $m=1$ is not possible. If $m\leq k$, then $\pi\in L^{(k)}_n$, so $\dist(\mathrm{id}_n,\pi)=1$, which agrees with the
bound.

Assume now that $m>k$. We will construct $\gamma\in L^{(k)}_n$ such that
$|D(\gamma^{-1}\pi,\mathrm{id})|\leq m-(k-1)$, and then invoke the inductive hypothesis.

Write $\pi$ as a product of disjoint cycles $C_1, \ldots, C_r$, and for each nontrivial cycle $C_i$ fix any cyclic ordering of its
elements. By listing the elements of the nontrivial cycles in these cyclic orders (cycle by cycle), we
obtain a sequence $v_1,\ldots,v_m$ that enumerates $M$ with the property that for every $1\leq j<m$,
either $\pi(v_j)=v_{j+1}$, or else $v_j$ is the last element of its cycle (in which case $\pi(v_j)$ is the
first element of that same cycle and hence appears earlier in the list).

Let $T:=\{v_1,\ldots,v_k\}$ be the first $k$ elements of the previously obtained sequence, and observe that for every $x\in T\setminus\{v_k\}$ we have
$\pi(x)\in T$: indeed, if $x=v_j$ with $j<k$, then either $\pi(x)=v_{j+1}\in T$, or $x$ is the last
element of its cycle and $\pi(x)$ is the first element of that cycle, which also lies in $T$.

Since $\pi$ is injective and $\pi(T\setminus\{v_k\})\subseteq T$, the set $\pi(T\setminus\{v_k\})$ has size
$k-1$, so there is a unique element $y\in T\setminus \pi(T\setminus\{v_k\})$. Define $\gamma\in S_n$ by
\[
\gamma(x)=
\begin{cases}
\pi(x) & \text{if } x\in T\setminus\{v_k\},\\
y & \text{if } x=v_k,\\
x & \text{otherwise.}
\end{cases}
\]
By construction, $\gamma$ only moves elements of $T$, so $\gamma\in L^{(k)}_n$. Moreover, for every
$x\in T\setminus\{v_k\}$ we have
\(
(\gamma^{-1}\pi)(x)=\gamma^{-1}(\pi(x))=\gamma^{-1}(\gamma(x))=x,
\)
so $\gamma^{-1}\pi$ fixes all elements of $T\setminus\{v_k\}$. Since $T\subseteq M$, these are $k-1$
indices that were moved by $\pi$, and thus
\(
|D(\gamma^{-1}\pi,\mathrm{id}_n)|\ \leq m-(k-1).
\)

Now, since $\pi=\gamma\cdot(\gamma^{-1}\pi)$ and $\gamma\in L^{(k)}_n$, we have
\[
\dist(\mathrm{id}_n,\pi)\ \leq 1+\dist(\mathrm{id}_n,\gamma^{-1}\pi)
\ \leq 1+\left\lceil \frac{(m-(k-1))-1}{k-1}\right\rceil
\ =\ \left\lceil \frac{m-1}{k-1}\right\rceil,
\]
where the second inequality is the inductive hypothesis. This completes the induction.

Finally, since $m\leq n$, we conclude that for every $\pi\in S_n$, we have
\(
\dist(\mathrm{id}_n,\pi)\ \leq \left\lceil \frac{n-1}{k-1}\right\rceil,
\)
and hence $\diam(\mathrm{Cay}(S_n, L^{(k)}_n)) \leq\left\lceil \frac{n-1}{k-1}\right\rceil$.
\end{proof}



\begin{proposition}[Upper bounds]\label{thm:static-upper}
In the static $\ell_k$-local model,
\[
S_{\ell}(n,k)\leq 28 n \lg n  \cdot \left\lceil\frac{n-1}{k-1}\right\rceil,
\]
and in the adaptive $\ell_k$-local model,
\[
A_{\ell}(n,k)\leq \frac{n^2 \lg n}{k}(1+o(1)).
\]
\end{proposition}
\begin{proof}
    Larcher, Martinsson, and Steger proved that $S(n)  \leq 28 n \lg n$~\cite{larcher2022solvingstaticpermutationmastermind}, and thus combining~\Cref{lemma:diam-prod} and~\Cref{lem:routing} directly yields the first result. 

    On the other hand, El Ouali et al.~\cite{ouali2016querycomplexityblackpegabmastermind} proved that $A(n) \leq (n-3) \lceil \lg n \rceil + \frac{5}{2}n - 1$. Again, combining~\Cref{lemma:diam-prod} and~\Cref{lem:routing} yields the result.  
\end{proof}
%

%% file: figures/k22.tex
\begin{tikzpicture}[
    scale=0.9,
    node distance=1.5cm,
    vertex/.style={circle, draw, fill=white, minimum size=6mm, font=\small},
    position/.style={circle, draw=blue!60!black, very thick, fill=blue!10, minimum size=8mm, inner sep=1pt},
    symbol/.style={circle, draw=green!60!black, very thick, fill=green!10, minimum size=8mm, inner sep=1pt},
    edge/.style={thick},
    untested/.style={very thick, black},
    tested/.style={thick, gray!40},
    k22edge/.style={very thick, red, line width=2pt}
]


\node[label, align=center] at (-5, 3.6) {$U_T$};

\node[position] (p1) at (-7, 3) {$1$};
\node[position] (p2) at (-7, 2) {$2$};
\node[position] (p3) at (-7, 1) {$3$};
\node[position] (p4) at (-7, 0) {$4$};
\node[position] (p5) at (-7, -1) {$5$};

\node[symbol] (s1) at (-3, 3) {$s_1$};
\node[symbol] (s2) at (-3, 2) {$s_2$};
\node[symbol] (s3) at (-3, 1) {$s_3$};
\node[symbol] (s4) at (-3, 0) {$s_4$};
\node[symbol] (s5) at (-3, -1) {$s_5$};

\draw[tested] (p1) -- (s2);
\draw[tested] (p1) -- (s4);
\draw[tested] (p2) -- (s1);
\draw[tested] (p2) -- (s5);
\draw[tested] (p3) -- (s2);
\draw[tested] (p3) -- (s5);
\draw[tested] (p4) -- (s1);
\draw[tested] (p5) -- (s3);
\draw[tested] (p5) -- (s5);

\draw[untested] (p1) -- (s1);
\draw[untested] (p1) -- (s5);
\draw[untested] (p2) -- (s2);
\draw[untested] (p4) -- (s5);
\draw[untested] (p5) -- (s1);
\draw[untested] (p5) -- (s2);

\draw[k22edge] (p2) -- (s3);
\draw[k22edge] (p2) -- (s4);
\draw[k22edge] (p4) -- (s3);
\draw[k22edge] (p4) -- (s4);

\node[position, fill=red!30, thick] (p2h) at (-7, 2) {$2$};
\node[position, fill=red!30, thick] (p4h) at (-7, 0) {$4$};
\node[symbol, fill=red!30, thick] (s3h) at (-3, 1) {$s_3$};
\node[symbol, fill=red!30, thick] (s4h) at (-3, 0) {$s_4$};


   \draw[lipicsYellow, line width=5.5pt, -latex] (-2.25, 1) -- (-0.7, 1);


\def\deltaV{-0.4}

\def\deltaL{-2.0}

\node[label, font=\bfseries] at (1.2+\deltaL, 3.3 + \deltaV) {$\sigma^\star$};
\node[position, minimum size=6mm] (i1) at (2+\deltaL, 3+ \deltaV) {$1$};
\node[position, minimum size=6mm] (i2) at (2+\deltaL, 2.2+ \deltaV) {$2$};
\node[position, minimum size=6mm] (i3) at (2+\deltaL, 1.4+ \deltaV) {$3$};
\node[position, minimum size=6mm] (i4) at (2+\deltaL, 0.6+ \deltaV) {$4$};
\node[position, minimum size=6mm] (i5) at (2+\deltaL, -0.2+ \deltaV) {$5$};

\node[symbol, minimum size=6mm] (t1) at (4+\deltaL, 3+ \deltaV) {$s_1$};
\node[symbol, minimum size=6mm] (t2) at (4+\deltaL, 2.2+ \deltaV) {$s_2$};
\node[symbol, minimum size=6mm] (t3) at (4+\deltaL, 1.4+ \deltaV) {$s_3$};
\node[symbol, minimum size=6mm] (t4) at (4+\deltaL, 0.6+ \deltaV) {$s_4$};
\node[symbol, minimum size=6mm] (t5) at (4+\deltaL, -0.2+ \deltaV) {$s_5$};

\draw[edge, blue!70!black, line width=1.2pt] (i1) -- (t1);
\draw[edge, orange!70!black, line width=1.2pt] (i2) -- (t3);
\draw[edge, blue!70!black, line width=1.2pt] (i3) -- (t2);
\draw[edge, orange!70!black, line width=1.2pt] (i4) -- (t4);
\draw[edge, blue!70!black, line width=1.2pt] (i5) -- (t5);

\node[position, minimum size=6mm, fill=red!30, thick] at (2+\deltaL, 2.2+ \deltaV) {$2$};
\node[position, minimum size=6mm, fill=red!30, thick] at (2+\deltaL, 0.6+ \deltaV) {$4$};
\node[symbol, minimum size=6mm, fill=red!30, thick] at (4+\deltaL, 1.4+ \deltaV) {$s_3$};
\node[symbol, minimum size=6mm, fill=red!30, thick] at (4+\deltaL, 0.6+ \deltaV) {$s_4$};

\node[label, font=\bfseries] at (5.2+\deltaL, 3.3+ \deltaV) {$\sigma^\dagger$};
\node[position, minimum size=6mm] (j1) at (6+\deltaL, 3+ \deltaV) {$1$};
\node[position, minimum size=6mm] (j2) at (6+\deltaL, 2.2+ \deltaV) {$2$};
\node[position, minimum size=6mm] (j3) at (6+\deltaL, 1.4+ \deltaV) {$3$};
\node[position, minimum size=6mm] (j4) at (6+\deltaL, 0.6+ \deltaV) {$4$};
\node[position, minimum size=6mm] (j5) at (6+\deltaL, -0.2+ \deltaV) {$5$};

\node[symbol, minimum size=6mm] (u1) at (8+\deltaL, 3+ \deltaV) {$s_1$};
\node[symbol, minimum size=6mm] (u2) at (8+\deltaL, 2.2+ \deltaV) {$s_2$};
\node[symbol, minimum size=6mm] (u3) at (8+\deltaL, 1.4+ \deltaV) {$s_3$};
\node[symbol, minimum size=6mm] (u4) at (8+\deltaL, 0.6+ \deltaV) {$s_4$};
\node[symbol, minimum size=6mm] (u5) at (8+\deltaL, -0.2+ \deltaV) {$s_5$};

\draw[edge, blue!70!black, line width=1.2pt] (j1) -- (u1);
\draw[edge, orange!70!black, line width=1.2pt] (j2) -- (u4);  
\draw[edge, blue!70!black, line width=1.2pt] (j3) -- (u2);
\draw[edge, orange!70!black, line width=1.2pt] (j4) -- (u3);  
\draw[edge, blue!70!black, line width=1.2pt] (j5) -- (u5);

\node[position, minimum size=6mm, fill=red!30, thick] at (6+\deltaL, 2.2+ \deltaV) {$2$};
\node[position, minimum size=6mm, fill=red!30, thick] at (6+\deltaL, 0.6+ \deltaV) {$4$};
\node[symbol, minimum size=6mm, fill=red!30, thick] at (8+\deltaL, 1.4+ \deltaV) {$s_3$};
\node[symbol, minimum size=6mm, fill=red!30, thick] at (8+\deltaL, 0.6+ \deltaV) {$s_4$};


\draw[dashed, thick, gray!50, rounded corners] (1.5+\deltaL, -0.7+ \deltaV) rectangle (4.5+\deltaL, 3.5+ \deltaV);
\draw[dashed, thick, gray!50, rounded corners] (5.5+\deltaL, -0.7+ \deltaV) rectangle (8.5+\deltaL, 3.5+ \deltaV);

\end{tikzpicture}

%% file: window-locality.tex
\section{Window Locality}

In this section we consider the $w_k$-local model. Intuitively, one would expect this more restrictive model to require many more guesses than the $\ell_k$-local model. However, perhaps surprisingly, showing a good separation result is one of the questions we will leave open.

Let us start by remarking that diameter-based upper bounds, as in~\Cref{sec:diameter-upper}, cannot help us here, since in the $w_k$-local model the diameter is $\Omega(n^2/k^2)$ as we show next. 

\begin{proposition}
    Let $W_k$ be the set of permutations $\omega \in S_n$ for which there exists an index $i$ such that $\omega(j) = j$ for $j \not \in \{i, i+1, \ldots, i+k-1\}$. Then, 
    \[
    \diam(\Cay(S_n, W_k)) \geq \frac{\lceil n^2 /2 \rceil}{k(k-1)} \geq \frac{1}{2}\left(\frac{n}{k}\right)^2.
    \]
\end{proposition}
\begin{proof}
    Consider the permutation $\pi^\dagger \colon i \mapsto n+1 - i$, and let us prove that $\dist(\pi^\dagger, \id_n) \geq \frac{\lceil n^2 /2 \rceil}{k(k-1)}$. For this, define the function $\delta \colon S_n \to \mathbb{N}$ by
    \(
    \delta(\pi) := \sum_{i=1}^n |\pi(i) - i|.
    \)
    Note that $\delta(\id_n) = 0$, whereas 
    \[
    \delta(\pi^\dagger) = \sum_{i=1}^n |n+1 - i -i| =\sum_{i=1}^n |n+1 - 2i| = \left \lceil \frac{n^2 }{2} \right \rceil,
    \]
    where the last equality can be seen by splitting the sum into two parts depending on whether the argument of the absolute value is positive or not.
    
    To conclude, it suffices to show that, for any $\pi \in S_n$ and $\omega \in W_k$, we have
    \begin{equation}\label{eq:delta}
    \delta(\omega  \pi ) - \delta(\pi) \leq  k(k-1).
     \end{equation}
    Indeed, if $\dist(\pi^\dagger, \id_n) < \frac{\lceil n^2 /2 \rceil}{k(k-1)}$, we could write $\pi^\dagger$ as $\omega_1 \omega_2 \dots \omega_t \id_n$ with $t < \frac{\lceil n^2 /2 \rceil}{k(k-1)}$ and $\omega_i \in W_k$ for all $i$. But then we would have
    \[
    \delta(\pi^\dagger) = \sum_{i=1}^t \delta(\omega_i \dots \omega_1 \id_n) - \delta(\omega_{i+1} \dots \omega_1 \id_n) \leq t \cdot k(k-1) < \lceil n^2 /2 \rceil,
    \]
    which contradicts the fact that $\delta(\pi^\dagger) = \lceil n^2 /2 \rceil$.   
    We now prove~\eqref{eq:delta}. let $I_\omega$ be the interval $\{i, \ldots, i+k-1\}$ from the definition of $W_k$ in which $\omega$ is potentially different from $\id_n$, and then observe that
    \begin{align*}
        \delta( \omega\pi) &= \sum_{i=1}^n |\omega(\pi(i)) - i|\\
        &= \sum_{i : \pi(i) \not\in I_\omega} | \pi(i) - i|  + \sum_{i : \pi(i) \in I_\omega} |\omega(\pi(i))-i|\\
        &= \sum_{i : \pi(i) \not\in I_\omega} | \pi(i) - i|  + \sum_{i : \pi(i) \in I_\omega} |(\omega(\pi(i))- \pi(i)) + (\pi(i) - i)|\\
        &\leq \sum_{i : \pi(i) \not\in I_\omega} | \pi(i) - i|  +  \sum_{i : \pi(i) \in I_\omega} |\omega(\pi(i))- \pi(i)| + \sum_{i : \pi(i) \in I_\omega} | \pi(i) - i| \\
        &= \delta(\pi)  +   \sum_{i : \pi(i) \in I_\omega} |\omega(\pi(i))- \pi(i)|.
    \end{align*}
    But by definition of $W_k$, we have that $\omega(I_\omega) = I_\omega$, and thus
    $|\omega(j) - j| \leq k-1$ for $j \in I_\omega$.  Therefore,
    \[
     \delta( \omega\pi) \leq \delta(\pi) + \sum_{i : \pi(i) \in I_\omega} (k-1) = \delta(\pi) + k(k-1),
    \]
    which proves~\eqref{eq:delta} and thus concludes the proof.
\end{proof}

Naturally, the diameter of the Cayley graph yields a lower bound on $A_w(n, k)$, but in this case such a bound is weaker than the one we get from the trivial inequality $A_\ell(n, k) \leq A_w(n, k)$, combined with~\Cref{thm:adaptive-ell}.

While using the diameter bound between guesses of a normal permutation Mastermind strategy would only yield an $O\left(\frac{n^3 \lg n}{k^2}\right)$ upper bound, we show that $O(n^2)$ is always a valid upper bound for $S_w(n, k)$, and it holds already for $k = 2$. Since the diameter lower bound immediately shows $S_w(n, 2) = \Omega(n^2)$, we focus on proving the upper bound of~\Cref{thm:window}.

\begin{proposition}
    We have $S_w(n, 2) = O(n^2)$.
\end{proposition}
\begin{proof}
    First, let us assume the codemaker is \emph{generous}, in the sense of~\Cref{sec:adaptive}, and thus reveals after every guess which of the $k=2$ changed positions is correct. We will then show how to get rid of this assumption with only a constant overhead.

    We will now use the following ``conveyor belt'' strategy (see~\Cref{alg:conveyor-belt}): take the sequence of $n-1$ adjacent transpositions $(1 \, 2), (2\, 3), \dots, (n-1 \; n)$ and repeat it $n$ times. For example for $n = 4$, this strategy guesses:
    \begin{align*}
1.\;&(\textcolor{red}{1}\,\textcolor{blue}{2}\,\textcolor{green!60!black}{3}\,4)
&\qquad
6.\;&(\textcolor{green!60!black}{3}\,4\,\textcolor{blue}{2}\,\textcolor{red}{1})\\
2.\;&(\textcolor{blue}{2}\,\textcolor{red}{1}\,\textcolor{green!60!black}{3}\,4)
&
7.\;&(\textcolor{green!60!black}{3}\,4\,\textcolor{red}{1}\,\textcolor{blue}{2})\\
3.\;&(\textcolor{blue}{2}\,\textcolor{green!60!black}{3}\,\textcolor{red}{1}\,4)
&
8.\;&(4\,\textcolor{green!60!black}{3}\,\textcolor{red}{1}\,\textcolor{blue}{2})\\
4.\;&(\textcolor{blue}{2}\,\textcolor{green!60!black}{3}\,4\,\textcolor{red}{1})
&
9.\;&(4\,\textcolor{red}{1}\,\textcolor{green!60!black}{3}\,\textcolor{blue}{2})\\
5.\;&(\textcolor{green!60!black}{3}\,\textcolor{blue}{2}\,4\,\textcolor{red}{1})
&
10.\;&(4\,\textcolor{red}{1}\,\textcolor{blue}{2}\,\textcolor{green!60!black}{3})
\end{align*}
    The nice property of this sequence of adjacent transpositions is that each element goes through every position, and thus, since the codemaker is generous, after this sequence of guesses, the codebreaker will know exactly what the secret permutation is. 
\begin{algorithm}[t]
\caption{\textsc{ConveyorBelt}$(n)$}
\label{alg:conveyor-belt}
\begin{algorithmic}[1]
\State $t \gets 1$
\State $\pi_t \gets \id$ \Comment{initial guess}
\For{$m = 1$ \textbf{to} $n$} \Comment{repeat a full left-to-right pass $n$ times}
    \State $\pi' \gets \pi_t$
    \For{$i = 1$ \textbf{to} $n-1$} \Comment{adjacent transposition $(i\;\; i+1)$}
        \State swap $\pi'(i)$ and $\pi'(i+1)$
        \State $\pi_{t+1} \gets \pi'$
        \State $t \gets t+1$
    \EndFor
\EndFor
\end{algorithmic}
\end{algorithm}
More formally, for each value $1 \leq m \leq n$ (line 3 of~\Cref{alg:conveyor-belt}), the first element of $\pi'$ right after line 4 is $m$, and it will be at position $2 \leq j \leq n$ after exactly $j-1$ guesses. Thus, the exact position of $m$ must be revealed by the generous codemaker on that pass.

After having identified the position of each element, the codebreaker knows the target permutation, and can reach the secret permutation through a shortest path. Using the well-known identity $\diam(S_n, W_2) = \binom{n}{2}$, the codebreaker will reach the secret permutation in $O(n^2)$ guesses.

Now, it only remains to show how to simulate the generous codemaker with only a constant factor overhead. For this, consider for example that for $n = 3$ the last guess was $(1 \, 2 \, 3)$ and we now want to guess $(2 \, 1 \, 3)$ but receiving generous feedback. We will do the sequence of guesses:
    \(
    (1 \, 2 \, 3) \to (1 \, 3 \, 2) \to (3 \, 1 \, 2) \to (3 \, 2\, 1) \to (2\, 3 \, 1)  \to (2\, 1\, 3).
    \)
    Intuitively, the scores received throughout this sequence will uniquely determine the ``generous'' feedback for the guess $(2 \, 1 \, 3)$.
    More in general, we will replace each transposition $(i \, i+1)$ by $6$ guesses that traverse the $6$ possible permutations of $\{i, i+1, i+2\}$ (or $\{i-1, i, i+1\}$ in case $i = n-1$). The fact that going through the $6$ permutations allows us to determine if any of the $3$ cycled elements should go into one of the $3$ positions of the cycle can be checked via finite computation. Indeed, the following Python3 code is enough:

\begin{lstlisting}[language=Python, style=mystyle]
import itertools

possible_masks = list(set(itertools.permutations([0,0,1,2,3], 3)))
guesses = list(itertools.permutations([1,2,3]))

results = {}
for pm in possible_masks:
    results[pm] = []
    for g in guesses:
        results[pm].append(sum(1 for m, g in zip(pm, g) if m == g))

for pm1, pm2 in itertools.combinations(possible_masks, 2):
    assert results[pm1] != results[pm2]
\end{lstlisting}
    

    To conclude the proof, we observe that the strategy is indeed static.
\end{proof}

%% file: complexity.tex
\section{The Complexity of Locality}

This section is dedicated to proving~\Cref{thm:np-hardness}. We will start by proving NP-hardness of the permutation variant without considering the locality restriction, and then show how a modification of the proof yields hardness even in the $\ell_3$-local setting. We will conclude the section by proving that in the $\ell_2$-local setting there is a randomized polynomial algorithm.

Our reduction is from \textsf{Monotone-1-in-3-SAT}, where the input is a 3-CNF formula with only positive literals, and the goal is to decide if some assignment of its variables assigns exactly 1 literal per clause to true.  A more formal description is presented below. Interestingly, the first proof of NP-completeness of Mastermind was a reduction from \textsf{1-in-3-SAT} (non-monotone)~\cite{de2004np}, but our construction is different.

\newcommand{\yes}{\textcolor{green!50!black}{\textsc{Yes}}}

\newcommand{\no}{\textcolor{red!60!black}{\textsc{No}}}

\begin{center}
\fbox{\begin{tabular}{lp{11cm}}
{\small PROBLEM} : & \textsf{Monotone-1-in-3-SAT} \\
{\small INPUT} : & A 3-CNF formula $\varphi := \bigwedge_{i=1}^m C_i$, over variables $x_1, \ldots, x_n$, with each clause $C_i := (x^{(i)}_1 \lor x^{(i)}_2 \lor x^{(i)}_3)$, and each $x^{(i)}_j \in \{x_1, \ldots, x_n\}$.
\\ 
{\small OUTPUT} : & {\yes}, if there is an assignment $\tau \colon \{x_1, \ldots, x_n\} \to \{0, 1\}$ such that for every clause $C_i$ we have $\tau\left(x^{(i)}_1\right) + \tau\left(x^{(i)}_2\right) + \tau\left(x^{(i)}_3\right)  = 1$. {\no} otherwise.
\end{tabular}}
\end{center}

\begin{theorem}[\cite{goldComplexityAutomatonIdentification1978}, see~\cite{DARMANN202145}]
    \textsf{Monotone-1-in-3-SAT} is $\mathrm{NP}$-complete.
\end{theorem}

Now we present the \textsf{Permutation-Mastermind-SAT} problem.

\begin{center}
\fbox{\begin{tabular}{lp{11cm}}
{\small PROBLEM} : & \textsf{Permutation-Mastermind-SAT} \\
{\small INPUT} : & A sequence of permutations $\pi_1, \ldots, \pi_T \in S_n$, and a sequence of black-peg scores $(b_1, \ldots, b_T) \in \{0, \ldots, n\}$.
\\ 
{\small OUTPUT} : & {\yes}, if there is some permutation $\sigma^\star \in S_n$ such that $b(\pi_t, \sigma^\star) = b_t$ for each $1 \leq t \leq T$.
{\no} otherwise.
\end{tabular}}
\end{center}

\begin{theorem}\label{thm:np-complete-1}
    \textsf{Permutation-Mastermind-SAT} is $\mathrm{NP}$-complete.
\end{theorem}
\begin{proof}
    Membership is trivial, since it suffices to guess $\sigma^\star \in S_n$. For hardness, we reduce from~\textsf{Monotone-1-in-3-SAT}. Let $\varphi := \bigwedge_{i=1}^m C_i$, over variables
    $x_1, \ldots, x_n$, be an input instance of~\textsf{Monotone-1-in-3-SAT}.
    Then, we set $N = 3n$, and we will construct an instance of~\textsf{Permutation-Mastermind-SAT} over $S_N$.

    The first guess will be simply the identity $\pi_1 := \id_N$, and the score is $b_1 := 0$. This enforces that the secret permutation $\sigma^\star$ holds $\sigma^\star(i) \neq i$ for each $1 \leq i \leq N$.
    In order to define the next guesses we will need some notation.
    For each $1 \leq i \leq n$, let us define the \emph{block} $B_i := (3i-2, 3i-1, 3i)$. 
    Intuitively, the first part of our construction will enforce that in the secret permutation $\sigma^\star = (s_1 \, s_2 \ldots s_N)$, we have $\{ s_{3i-2}, s_{3i-1}, s_{3i} \} = \{{3i-2}, {3i-1}, {3i}\}$, meaning that the secret permutation is only permuting within blocks of 3 elements.
    
    Using notation $\oplus$ for concatenation, we can define the identity permutation $\id_N \in S_N$ as $B_1 \oplus B_2 \dots \oplus B_n$.  
    Now, for each $1 \leq i < j \leq n$ and $\sigma \in S_3$, we define the permutation 
    \[
    \gamma_{i,j, \sigma} := B_1 \oplus \dots \oplus B_{i-1} \oplus \sigma(B_j) \oplus \dots \oplus B_{j-1}  \oplus \sigma(B_i) \oplus \dots \oplus B_n.
    \]
    For example, if $n = 4$, and $\sigma = (1 \, 3 \, 2)$, then
    \[
   \gamma_{1,3, \sigma} = \sigma(B_3) \oplus B_2 \oplus \sigma(B_1) \oplus B_4
= (7\,9\,8\,4\,5\,6\,1\,3\,2\,10\,11\,12).
    \]
    Let $\sigma_1, \ldots, \sigma_6$ be an arbitrary enumeration of $S_3$. We now introduce $6\binom{n}{2}$ new guesses, corresponding to the permutations $\gamma_{i, j, \sigma_k}$ for $1 \leq i < j \leq n$ and $1 \leq k \leq 6$. The score for each of them will be $0$.

\begin{claim}\label{claim:blocks}
    Any permutation $\sigma^\star$ compatible with the scores for the $1 + 6\binom{n}{2}$ guesses made thus far satisfies 
    \[
    \forall 1 \leq i \leq n, \quad\{\sigma^\star(3i-2), \sigma^\star(3i-1), \sigma^\star(3i)\} = \{3i-2, 3i-1, 3i\}.
    \]
\end{claim}
\begin{claimproof}
    Assume the condition fails for some $1 \leq i \leq n$. Then, some element of $\{\sigma^\star(3i-2), \sigma^\star(3i-1), \sigma^\star(3i)\}$ belongs to a block $B_j$ for $j \neq i$. Let us assume $i < j$.\footnote{Otherwise we implicitly do $i := \min(i, j), j = \max(i, j)$.} Then, for some $k$ it must be the case that guess $\gamma_{i, j, \sigma_k}$ gets score at least $1$, which contradicts the compatibility assumption with the received scores which are all $0$.
\end{claimproof}

We therefore have that the secret permutation $\sigma^\star$ must be of the form
\[
\sigma^{(1)}(B_1) \oplus \sigma^{(2)}(B_2) \oplus \dots \oplus \sigma^{(n)}(B_n),
\]
where each $\sigma^{(i)}$ is a permutation with no fixed points. In particular, the only two possibilities in $S_3$ are $\alpha := (2\,3\,1)$ and $\beta := (3\,1\,2)$. Intuitively, we will have $\sigma^{(i)} = \alpha$ if variable $x_i$ should be assigned to true in $\varphi$, and $\beta$ if it should be assigned to false. 

Concretely, for each clause $C := (x_a \lor x_b \lor x_c)$ with $a < b < c$, we create a permutation 
\[
\pi_C := B_1 \oplus \dots \oplus B_{a-1} \oplus \alpha(B_a) \oplus B_{a+1} \oplus \dots \oplus \alpha(B_b) \oplus \dots \oplus \alpha(B_c) \oplus \dots \oplus B_n.
\]
The score given to each of these guesses will be $3$.

We are now ready to prove the correctness of the reduction.
Let $\pi_1, \dots, \pi_T$ with $T = 1 + 6\binom{n}{2} + m$, and $b_1, \dots, b_T$ be the sequence of guesses and scores of our construction.

\begin{claim}
    If $\varphi$ is a \yes-instance for \textsf{Monotone-1-in-3-SAT}, then $(\pi_1, \dots, \pi_T), (b_1, \dots, b_T)$ is a \yes-instance for \textsf{Permutation-Mastermind-SAT}.
\end{claim}
\begin{claimproof}
    Let $\tau$ be the satisfying assignment for $\varphi$ which witnesses its membership in~\textsf{Monotone-1-in-3-SAT}.
    Then, let $\sigma^{\star}$ be the permutation $\sigma^{(1)}(B_1) \oplus \sigma^{(2)}(B_2) \oplus \dots \oplus \sigma^{(n)}(B_n)$ where \[
    \sigma^{(i)}  =\begin{cases}
      \alpha & \text{if } \tau(x_i) = 1\\
      \beta & \text{otherwise.}
    \end{cases}
    \]
    It suffices to check that all scores on $\sigma^\star$ are as constructed. First, since both $\alpha$ and $\beta$ have no fixed points, we indeed have $b(\id_N, \sigma^\star) = 0$. Let us write $\sigma[i] := (\sigma(3i-2),\sigma(3i-1), \sigma(3i)) $ for the $i$-th block of a permutation $\sigma$.
    
    Then, we claim that $b(\gamma_{i, j, \sigma}, \sigma^\star) = 0$ for every $1 \leq i < j \leq n$ and $\sigma \in S_3$. Indeed, for $k \in \{1, \ldots, n\} \setminus \{i, j\}$, we have $\gamma_{i, j, \sigma}[k] = B_k$, whereas $\sigma^{\star}[k]$ is either $\alpha(B_k)$ or $\beta(B_k)$, and in either case $b(\gamma_{i, j, \sigma}[k], \sigma^\star[k]) = 0$. For $k \in \{i, j\}$, we also have that $\sigma^{\star}[k]$ is either $\alpha(B_k)$ or $\beta(B_k)$, whereas $\gamma_{i, j, \sigma}[k] = \sigma(B_{\{i, j\} \setminus\{k\}})$, and thus $b(\gamma_{i, j, \sigma}[k], \sigma^\star[k]) = 0$.
    
    Finally, for each clause $C$, we claim that $b(\pi_C, \sigma^\star) = 3$. Indeed, by assumption, there is exactly one variable $x_i \in C$ such that $\tau(x_i) = 1$, so $\sigma^{(i)} = \alpha$, and thus $\pi_C[i] = \alpha(B_i)$ and by construction $\sigma^\star[i] = \alpha(B_i)$, so $b(\pi_C[i], \sigma^{\star}[i]) = 3$. If the other variables of $C$ are $x_j$ and $x_k$, then for any $\ell \not\in \{i, j, k\}$ we have $\pi_C[\ell] = B_\ell$, whereas $\sigma^\star[\ell]$ is either $\alpha(B_\ell)$ or $\beta(B_\ell)$, and in either case $b(\pi_C[\ell], \sigma^\star[\ell]) = 0$. For $\ell \in \{j, k\}$ we have $\tau(x_\ell) = 0$ by assumption, and thus $\sigma^\star[\ell] = \beta(B_\ell)$, whereas $\pi_C[\ell] = \alpha(B_\ell)$, and again $b(\pi_C[\ell], \sigma^\star[\ell]) = 0$.
\end{claimproof}

To complete our proof, we prove the other direction of the reduction.
\begin{claim}
    If $(\pi_1, \dots, \pi_T), (b_1, \dots, b_T)$ is a \yes-instance for \textsf{Permutation-Mastermind-SAT}, then $\varphi$ is a \yes-instance for \textsf{Monotone-1-in-3-SAT}.
\end{claim}
\begin{claimproof}
    Assume there is some  $\sigma^\star$ compatible with the transcript $(\pi_1, \dots, \pi_T), (b_1, \dots, b_T)$. Then, based on~\Cref{claim:blocks}, we have for each $1 \leq i\leq n$ that $\sigma^\star[i] = \sigma^{(i)}(B_i)$ for some $\sigma^{(i)} \in S_3$. Because $\pi_1 = \id_N$ and $b_1 = 0$, all the $\sigma^{(i)}$ must have no fixed points. But on $S_3$ the only two permutations without fixed points are exactly $\alpha := (2\,3\,1)$ and $\beta := (3\,1\,2)$. Therefore, $\sigma^{(i)} \in \{\alpha, \beta\}$ for each $1 \leq i \leq n$. Now, define the assignment $\tau \colon \{x_1, \ldots, x_n\} \to \{0, 1\}$ by 
    \[
    \tau(x_i) := \begin{cases}
        1 &\text{ if } \sigma^{(i)} = \alpha\\
        0 & \text{ otherwise}.
    \end{cases}
    \]
    Now, let $C := (x_a \lor x_b \lor x_c)$ be an arbitrary clause in $\varphi$. Then, since guess $\pi_C$ got score $3$, we have that
    \begin{equation}\label{eq:score3}
    b(\alpha(B_a), \sigma^{\star}[a]) +  b(\alpha(B_b), \sigma^{\star}[b]) +  b(\alpha(B_c), \sigma^{\star}[c]) = 3.
      \end{equation}
    However, for each $g \in \{a,b,c\}$, we have $\sigma^{\star}[g] = \alpha(B_g)$, in which case $b(\alpha(B_g), \sigma^{\star}[g]) = 3$ or $\sigma^{\star}[g] = \beta(B_g)$, in which case $b(\alpha(B_g), \sigma^{\star}[g]) = 0$. Thus, the only way to satisfy~\Cref{eq:score3} is that for exactly one $g \in \{a,b, c\}$, $\sigma^{(g)} = \alpha$, and thus $\tau$ satisfies exactly one variable of $C$.
\end{claimproof}

Since the reduction clearly takes polynomial time, this completes the proof.
\end{proof}

We will now show that the construction above can be adapted so that each pair of consecutive guesses differs in a fixed number of positions. The main difficulty is that, while the sequence of guesses $\pi_1, \dots, \pi_T$ from the proof of~\Cref{thm:np-complete-1} could be ``interpolated'' by a series of transpositions between each $\pi_t$ and $\pi_{t+1}$, the reduction needs to construct scores for these intermediate guesses as well, and it must do so in such a way that does not require knowing what a satisfying assignment (if there is any) looks like.

\begin{center}
\fbox{\begin{tabular}{lp{11cm}}
{\small PROBLEM} : & $k$-\textsf{Local-PM-SAT} \\
{\small INPUT} : & A sequence of permutations $\pi_1, \ldots, \pi_T \in S_n$, where each $\pi_{t+1} := \pi_t \circ \omega_{t+1}$, for some permutation $\omega_{t+1}$ of support size at most $k$,
and a sequence of black-peg scores $(b_1, \ldots, b_T) \in \{0, \ldots, n\}$.
\\ 
{\small OUTPUT} : & {\yes}, if there is some permutation $\sigma^\star \in S_n$ such that $b(\pi_t, \sigma^\star) = b_t$ for each $1 \leq t \leq T$.
{\no} otherwise.
\end{tabular}}
\end{center}

\begin{theorem}
    3-\textsf{Local-PM-SAT} is $\mathrm{NP}$-complete.
\end{theorem}

\begin{proof}
Membership is again trivial since it suffices to guess $\sigma^\star \in S_n$. 
For hardness, the proof is very similar to that of~\Cref{thm:np-complete-1}, also via a reduction from \textsf{Monotone-1-in-3-SAT}, but we will need some non-trivial modifications. Once again, if the input instance $\varphi$ has $n$ variables, we will construct an instance of 3-\textsf{Local-PM-SAT} over $S_N$ for $N := 3n$. 
The first guess will again be the identity $\id_N$ with a score of $0$, implying the secret permutation must not have any fixed points. Then, similarly to the proof of~\Cref{thm:np-complete-1}, we will construct guesses and scores enforcing that 
    \begin{equation}\label{eq:block-preservation}
    \forall 1 \leq i \leq n, \quad\{\sigma^\star(3i-2), \sigma^\star(3i-1), \sigma^\star(3i)\} = \{3i-2, 3i-1, 3i\}.
    \end{equation}
This part can be done in fact under the stronger $\ell_2$-local restriction. 
It suffices to take every pair of indices $1 \leq i < j \leq 3n$ such that $\lceil i/3\rceil \neq \lceil j/3\rceil$ (i.e., $i$ and $j$ belong to different blocks of size $3$), and guess the permutation $\delta_{i, j}$ defined by
\[
\delta_{i, j}(k) = \begin{cases}
    k & \text{if } k \not\in \{i, j\}\\
    i & \text{if } k = j\\
    j & \text{if } k = i.
\end{cases}
\]
The associated score of each guess $\delta_{i, j}$ will be $0$, and after each such guess, we will guess again the identity permutation $\id_N$ with a score of $0$, before making the next guess $\delta_{i', j'}$. This stage makes $O(n^2)$ guesses, and indeed guarantees~\Cref{eq:block-preservation} since any position in $B_i := \{3i-2, 3i-1, 3i\}$ is guessed in every position outside of $B_i$ and receives a score of $0$.
Now, we turn our attention to the \emph{clause guesses}. Let $C = (x_i \lor x_j \lor x_k)$ be a clause in $\varphi$, with $i < j < k$. 
Our goal is to make the guess
\[
\pi_C := B_1 \oplus \dots \oplus B_{i-1} \oplus \alpha(B_i) \oplus B_{i+1} \oplus \dots \oplus \alpha(B_j) \oplus \dots \oplus \alpha(B_k) \oplus \dots \oplus B_n.
\]
However, we will have to ``walk'' to such a guess by permutations of support size at most $3$. Let us show an example for $n = 4$, with the clause $C = (x_1 \lor x_3 \lor x_4)$.  For ease of notation, we will identify the identity permutation with $b_1 b_2 \dots b_{12}$, and highlight in blue the permuted elements of each guess.
\begin{align*}
b_{1}  b_{2}  b_{3} \;b_{4}  b_{5}  b_{6} \;b_{7}  b_{8}  b_{9} \;b_{10}  b_{11}  b_{12} \\
\textcolor{blue}{ b_{7} }  b_{2}  b_{3} \;b_{4}  b_{5}  b_{6} \;\textcolor{blue}{ b_{10} }  b_{8}  b_{9} \;\textcolor{blue}{ b_{1} }  b_{11}  b_{12} \\
b_{7}  \textcolor{blue}{ b_{8} }  b_{3} \;b_{4}  b_{5}  b_{6} \;b_{10}  \textcolor{blue}{ b_{11} }  b_{9} \;b_{1}  \textcolor{blue}{ b_{2} }  b_{12} \\
b_{7}  b_{8}  \textcolor{blue}{ b_{9} } \;b_{4}  b_{5}  b_{6} \;b_{10}  b_{11}  \textcolor{blue}{ b_{12} } \;b_{1}  b_{2}  \textcolor{blue}{ b_{3} } \\
\textcolor{blue}{ b_{8} }  \textcolor{blue}{ b_{9} }  \textcolor{blue}{ b_{7} } \;b_{4}  b_{5}  b_{6} \;b_{10}  b_{11}  b_{12} \;b_{1}  b_{2}  b_{3} \\
b_{8}  b_{9}  b_{7} \;b_{4}  b_{5}  b_{6} \;\textcolor{blue}{ b_{11} }  \textcolor{blue}{ b_{12} }  \textcolor{blue}{ b_{10} } \;b_{1}  b_{2}  b_{3} \\
b_{8}  b_{9}  b_{7} \;b_{4}  b_{5}  b_{6} \;b_{11}  b_{12}  b_{10} \;\textcolor{blue}{ b_{2} }  \textcolor{blue}{ b_{3} }  \textcolor{blue}{ b_{1} } \\
\textcolor{blue}{ b_{2} }  b_{9}  b_{7} \;b_{4}  b_{5}  b_{6} \;\textcolor{blue}{ b_{8} }  b_{12}  b_{10} \;\textcolor{blue}{ b_{11} }  b_{3}  b_{1} \\
b_{2}  \textcolor{blue}{ b_{3} }  b_{7} \;b_{4}  b_{5}  b_{6} \;b_{8}  \textcolor{blue}{ b_{9} }  b_{10} \;b_{11}  \textcolor{blue}{ b_{12} }  b_{1} \\
b_{2}  b_{3}  \textcolor{blue}{ b_{1} } \;b_{4}  b_{5}  b_{6} \;b_{8}  b_{9}  \textcolor{blue}{ b_{7} } \;b_{11}  b_{12}  \textcolor{blue}{ b_{10} } 
\end{align*}
Naturally, this sequence of permutations can be inverted to obtain back the identity. More formally, let us introduce notation $\alpha(i, j, k)$ for the permutation that applies $\alpha \in S_3$ to positions $i, j, k$, and similarly, we define $\beta(i, j, k)$. Then, the sequence of guesses corresponding to a clause $C = (x_i \lor x_j \lor x_k)$ is:
\begin{enumerate}
    \item $\alpha(3i-2, 3j-2, 3k-2)$,
    \item $\alpha(3i-1, 3j-1, 3k-1)$,
    \item $\alpha(3i, 3j, 3k)$,
    \item $\alpha(3i-2, 3i-1, 3i)$,
    \item $\alpha(3j-2, 3j-1, 3j)$,
    \item $\alpha(3k-2, 3k-1, 3k)$,
    \item $\beta(3i-2, 3j-2, 3k-2)$,
    \item $\beta(3i-1, 3j-1, 3k-1)$,
    \item $\beta(3i, 3j, 3k)$.
\end{enumerate}

Their corresponding scores are: $(0, 0, 0, 0, 0, 0, 1, 2, 3)$. Consecutive ``clause gadgets'' can be chained by appending the inverse walk back to $\id_N$ between clauses, with the corresponding scores.
We now make the main claim for the correctness of the reduction.

\begin{claim}\label{claim:final-np}
    A permutation $\sigma^\star$ is compatible with the scores constructed thus far if and only if we have
    \[
    \sigma^\star = \sigma^{(1)}(B_1) \oplus \sigma^{(2)}(B_2) \oplus \dots \oplus \sigma^{(n)}(B_n)
    \]
    where each $\sigma^{(i)}$ is either $\alpha := (2 \, 3\, 1)$ or $\beta := (3\, 1\, 2)$, and for each clause $(x_i \lor x_j \lor x_k) \in \varphi$, exactly one of $\sigma^{(i)}, \sigma^{(j)}, \sigma^{(k)}$ equals $\alpha$.
\end{claim}
\begin{claimproof}
    The first part is direct from~\Cref{eq:block-preservation} and the fact that the identity guess $\id_N$ has a score of $0$, since $\alpha$ and $\beta$ are the only permutations in $S_3$ without fixed points.

    The second part can be checked mechanically. It suffices to consider clause $(x_1 \lor x_2 \lor x_3)$ without loss of generality, and for each of the $8$ possibilities for $(\sigma^{(1)},\sigma^{(2)}, \sigma^{(3)})$, observe that the only ones that give the scores $(0,0,0,0,0,0,1,2,3)$ are those with exactly one $\alpha$. This is shown in~\Cref{tab:guesses-results}, and concludes the proof of the claim.
    \begin{table}
      \caption{Correctness of the sequence of clause guesses.}
        \label{tab:guesses-results}
        \centering
        \resizebox{\textwidth}{!}{%
        \begin{tabular}{
    c|
    *{8}{c}}
        \toprule
    $(\sigma^{(1)},\sigma^{(2)}, \sigma^{(3)})$ &   \small \textcolor{red}{$(\alpha, \alpha, \alpha)$} &  \small \textcolor{red}{$ (\alpha, \alpha, \beta)$} &   \small \textcolor{red}{$ (\alpha, \beta, \alpha)$} &  \small \textcolor{green!70!black}{$(\alpha, \beta, \beta)$} & \small  \textcolor{red}{$(\beta, \alpha, \alpha)$} &  \small    \textcolor{green!70!black}{$(\beta, \alpha, \beta)$} &  \small \textcolor{green!70!black}{$(\beta, \beta, \alpha)$}  &  \small \textcolor{red}{$(\beta, \beta, \beta)$} \\ \midrule
$b_{1}  b_{2}  b_{3} \;b_{4}  b_{5}  b_{6} \;b_{7}  b_{8}  b_{9} $& 0  & 0  & 0  & 0  & 0  & 0  & 0  & 0  \\
$\textcolor{blue}{ b_{4} }  b_{2}  b_{3} \;\textcolor{blue}{ b_{7} }  b_{5}  b_{6} \;\textcolor{blue}{ b_{1} }  b_{8}  b_{9} $& 0  & 0  & 0  & 0  & 0  & 0  & 0  & 0  \\
$b_{4}  \textcolor{blue}{ b_{5} }  b_{3} \;b_{7}  \textcolor{blue}{ b_{8} }  b_{6} \;b_{1}  \textcolor{blue}{ b_{2} }  b_{9} $& 0  & 0  & 0  & 0  & 0  & 0  & 0  & 0  \\
$b_{4}  b_{5}  \textcolor{blue}{ b_{6} } \;b_{7}  b_{8}  \textcolor{blue}{ b_{9} } \;b_{1}  b_{2}  \textcolor{blue}{ b_{3} } $& 0  & 0  & 0  & 0  & 0  & 0  & 0  & 0  \\
$\textcolor{blue}{ b_{5} }  \textcolor{blue}{ b_{6} }  \textcolor{blue}{ b_{4} } \;b_{7}  b_{8}  b_{9} \;b_{1}  b_{2}  b_{3} $& 0  & 0  & 0  & 0  & 0  & 0  & 0  & 0  \\
$b_{5}  b_{6}  b_{4} \;\textcolor{blue}{ b_{8} }  \textcolor{blue}{ b_{9} }  \textcolor{blue}{ b_{7} } \;b_{1}  b_{2}  b_{3} $& 0  & 0  & 0  & 0  & 0  & 0  & 0  & 0  \\
$b_{5}  b_{6}  b_{4} \;b_{8}  b_{9}  b_{7} \;\textcolor{blue}{ b_{2} }  \textcolor{blue}{ b_{3} }  \textcolor{blue}{ b_{1} } $& 0  & 0  & 0  & 0  & 0  & 0  & 0  & 0  \\
$\textcolor{blue}{ b_{2} }  b_{6}  b_{4} \;\textcolor{blue}{ b_{5} }  b_{9}  b_{7} \;\textcolor{blue}{ b_{8} }  b_{3}  b_{1} $& 3  & 2  & 2  & 1  & 2  & 1  & 1  & 0  \\
$b_{2}  \textcolor{blue}{ b_{3} }  b_{4} \;b_{5}  \textcolor{blue}{ b_{6} }  b_{7} \;b_{8}  \textcolor{blue}{ b_{9} }  b_{1} $& 6  & 4  & 4  & 2  & 4  & 2  & 2  & 0  \\
$b_{2}  b_{3}  \textcolor{blue}{ b_{1} } \;b_{5}  b_{6}  \textcolor{blue}{ b_{4} } \;b_{8}  b_{9}  \textcolor{blue}{ b_{7} } $& 9  & 6  & 6  & 3  & 6  & 3  & 3  & 0  \\
        \bottomrule
        \end{tabular}%
        }
    \end{table}
\end{claimproof}

Having proven~\Cref{claim:final-np}, we will proceed to prove the correctness of the reduction.

\begin{claim}
    If $\varphi$ is a \yes-instance for \textsf{Monotone-1-in-3-SAT}, then the sequence of guesses and scores $(\pi_1, \dots, \pi_T), (b_1, \dots, b_T)$ constructed above is a \yes-instance for 3-\textsf{Local-PM-SAT}.
\end{claim}
\begin{claimproof}
    Let $\tau$ be the satisfying assignment for $\varphi$ which witnesses its membership in~\textsf{Monotone-1-in-3-SAT}.
    Then, let $\sigma^{\star}$ be the permutation 
    \[ 
    \sigma^{(1)}(B_1) \oplus \sigma^{(2)}(B_2) \oplus \dots \oplus \sigma^{(n)}(B_n)  \] where \(
    \sigma^{(i)}  =\begin{cases}
      \alpha & \text{if } \tau(x_i) = 1\\
      \beta & \text{otherwise.}
    \end{cases}
    \)
    Then, by~\Cref{claim:final-np} we have that $\sigma^{\star}$ is compatible with $(\pi_1, \dots, \pi_T), (b_1, \dots, b_T)$, thus proving it is a \yes-instance for 3-\textsf{Local-PM-SAT}.
\end{claimproof}

\begin{claim}
   If $(\pi_1, \dots, \pi_T), (b_1, \dots, b_T)$ is a \yes-instance for 3-\textsf{Local-PM-SAT}, then $\varphi$ is a \yes-instance for \textsf{Monotone-1-in-3-SAT}.
\end{claim}
\begin{claimproof}
    Directly by~\Cref{claim:final-np} it suffices to take the assignment $\tau$ that assigns each $x_i$ to $1$ if $\sigma^{(i)}=\alpha$, and to $0$ otherwise.
\end{claimproof}

Since the reduction can clearly be constructed in polynomial time, we conclude the proof.
\end{proof}

Finally, we show that 2-\textsf{Local-PM-SAT} can be solved in randomized polynomial time, by formulating it as a search over perfect matchings with parity constraints. In particular, we consider the following problem.

\begin{center}
\fbox{\begin{tabular}{lp{11cm}}
{\small PROBLEM} : & $t$-\textsf{Dimensional Parity Perfect Matching} \\
{\small INPUT} : & A $G = (V, E)$, a weight function $w \colon E \to \mathbb{Z}^{\geq 0}$, a function $\gamma \colon E \to \mathbb{F}^t_2$,  a constant $\mathbf{c} \in \mathbb{F}^t_2$, and a number $r \in \mathbb{Z}$.
\\ 
{\small OUTPUT} : & A perfect matching $M$ of $G$ such that $w(M) := \sum_{e \in M} w(e) = r$, and such that $\sum_{e \in M} \gamma(e) = \mathbf{c}$.
\end{tabular}}
\end{center}

\begin{theorem}[\!\!\!{{\cite[Thm. 5.6]{Geelen2017}}}]\label{thm:dim-matching}
    There is a randomized algorithm that solves the $t$-\textsf{Dimensional Parity Perfect Matching} problem, over instances on $n$ vertices, in time $O(n^6 \lg^2 n \cdot \max_{e \in E(G)} w(e))$, with probability of success $1 - o(1)$.
\end{theorem}

Strictly speaking, the algorithm of Geelen and Kapadia~\cite{Geelen2017} is intended for minimizing $w(M)$, but this is only done in step 3 of their evaluation algorithm, where the minimum exponent of the variable $z$ is considered; by considering the term including $z^r$ instead, from the same polynomial, we obtain the form of~\Cref{thm:dim-matching}.



\begin{theorem}\label{thm:2local-rp}
\textnormal{2-Local-PM-SAT} is in $\mathsf{RP}$, and in particular, can be solved in time $O(n^7 \lg^2 n)$ with success probability $1-o(1)$.
\end{theorem}

\begin{proof}
Let $(\pi_1,\ldots,\pi_T)$ and $(b_1,\ldots,b_T)$ be an instance of \textnormal{2-Local-PM-SAT} over $S_n$.
Then, as in~\Cref{sec:adaptive}, we identify a permutation $\sigma\in S_n$ with the perfect matching
\(
B_\sigma \;:=\; \bigl\{\{i,\sigma(i)\} : i\in[n]\bigr\}
\)
in the complete bipartite graph $K_{n,n}$, and note that
$b(\pi_t,\sigma)=|B_{\pi_t}\cap B_\sigma|$.

We can assume without loss of generality that the transcript contains no consecutive guesses being equal, i.e., $\pi_{t+1} \neq \pi_{t}$, since otherwise we can remove the consecutive duplicates from the transcript without altering satisfiability.
Because of the $\ell_2$-locality restriction, for each $t\in[T-1]$ there is a (unique) transposition $\omega_{t+1}=(a_t\,b_t)$ such that
$\pi_{t+1}=\pi_t\circ\omega_{t+1}$.
Let
\(
u_t:=\pi_t(a_t),\) and \(v_t:=\pi_t(b_t),
\)
and define the four edges
\[
e_t^1:=\{a_t,u_t\},\quad e_t^2:=\{b_t,v_t\},\quad f_t^1:=\{a_t,v_t\},\quad f_t^2:=\{b_t,u_t\}.
\]
Then $B_{\pi_t} \symdiff  B_{\pi_{t+1}}=\{e_t^1,e_t^2,f_t^1,f_t^2\}$, and only these two positions
can change their contribution to the black-peg score when passing from $\pi_t$ to $\pi_{t+1}$.
Let $\Delta_t:=b_{t+1}-b_t$. Letting $B_{\sigma^\star}$ be the matching associated to a secret permutation $\sigma^\star$ compatible with the scores (if any exists), shows:
\begin{align*}
\Delta_t=2  &\iff \{f_t^1,f_t^2\}\subseteq B_{\sigma^\star},\\
\Delta_t=-2 &\iff \{e_t^1,e_t^2\}\subseteq B_{\sigma^\star},\\
\Delta_t=1  &\iff \bigl|B_{\sigma^\star}\cap\{f_t^1,f_t^2\}\bigr|=1\ \text{ and }\ B_{\sigma^\star}\cap\{e_t^1,e_t^2\}=\emptyset,\\
\Delta_t=-1 &\iff \bigl|B_{\sigma^\star}\cap\{e_t^1,e_t^2\}\bigr|=1\ \text{ and }\ B_{\sigma^\star}\cap\{f_t^1,f_t^2\}=\emptyset,\\
\Delta_t=0  &\iff B_{\sigma^\star}\cap\{e_t^1,e_t^2,f_t^1,f_t^2\}=\emptyset.
\end{align*}
Consequently, we will construct an instance of perfect matching with the following constraint sets: a set $\mathcal{F}$ of edges that must necessarily be present in the secret perfect matching $B_{\sigma^\star}$, a set $\mathcal{N}$ of edges that must necessarily \emph{not} be present in the secret perfect matching $B_{\sigma^\star}$, and a set $\mathcal{C}$ of pairs of edges such that for every pair $C \in \mathcal{C}$, exactly one of the edges in $C$ must be present in the secret perfect matching $B_{\sigma^\star}$. These sets of constraints are constructed by~\Cref{alg:construct}.

\begin{algorithm}
\caption{Construct Constraint Sets $\mathcal{C}, \mathcal{F}, \mathcal{N}$}\label{alg:construct}
\begin{algorithmic}[1]
    \State $\mathcal{C} \gets \emptyset, \mathcal{F} \gets \emptyset, \mathcal{N} \gets \emptyset$
    
    \For{$t = 1$ \textbf{to} $T-1$}
        \If{$\Delta_t = 2$}
            \State $\mathcal{F} \gets \mathcal{F} \cup \{f_t^1, f_t^2\}$
            \State $\mathcal{N} \gets \mathcal{N} \cup \{e_t^1, e_t^2\}$
        \ElsIf{$\Delta_t = -2$}
            \State $\mathcal{F} \gets \mathcal{F} \cup \{e_t^1, e_t^2\}$
            \State $\mathcal{N} \gets \mathcal{N} \cup \{f_t^1, f_t^2\}$
        \ElsIf{$\Delta_t = 1$}
            \State $\mathcal{N} \gets \mathcal{N} \cup \{e_t^1, e_t^2\}$
            \State $\mathcal{C} \gets \mathcal{C} \cup \{\{f_t^1, f_t^2\}\}$
        \ElsIf{$\Delta_t = -1$}
            \State $\mathcal{N} \gets \mathcal{N} \cup \{f_t^1, f_t^2\}$
            \State $\mathcal{C} \gets \mathcal{C} \cup \{\{e_t^1, e_t^2\}\}$
        \ElsIf{$\Delta_t = 0$}
            \State $\mathcal{N} \gets \mathcal{N} \cup \{f_t^1, f_t^2, e_t^1, e_t^2\}$
        \EndIf
    \EndFor
\end{algorithmic}
\end{algorithm}

By the previous analysis, we have the following observation.
\begin{observation}\label{obs:pm-inters}
    There exists a secret permutation $\sigma^\star$ compatible with the  transcript $(\pi_1, \dots, \pi_T), (b_1,\dots, b_T)$ if and only if $K_{n,n}$ has a perfect matching $M^\star$ such that: 
    \begin{enumerate}
        \item $| M^\star \cap B_{\pi_1}| = b_1$,
         \item  $|M^\star \cap \mathcal{N}| = 0$,
        \item for every $e \in \mathcal{F}$, $|M^\star \cap \{e\}| = 1$,
        \item for every pair of edges $C  \in \mathcal{C}$, $|M^\star \cap C| = 1$.
    \end{enumerate}
\end{observation}

We will now show how to decide the existence of such a matching $M^\star$ by reducing to the $t$-\textsf{Dimensional Parity Perfect Matching} problem. First, we define the edge weight function $w \colon E(K_{n, n}) \to \mathbb{Z}^{\geq 0}$ as follows:
\[
w(e) := \begin{cases}
    n & \text{if } e \in \mathcal{N},\\
    1 & \text{if } e \not\in \mathcal{N} \text{ and } e \not\in B_{\pi_1},\\
    0 & \text{otherwise}.
\end{cases}
\]
Now, note that, for any perfect matching $M$ of $K_{n,n}$, we have
\[
w(M) := \sum_{e \in M} w(e) = n\cdot|M \cap \mathcal{N}| + |(M \setminus \mathcal{N}) \setminus B_{\pi_1}|,
\]
    from where (i) if $|M \cap \mathcal{N}| > 0$, then $w(M) \geq n$, and (ii) if $|M \cap \mathcal{N}| = 0$, then $w(M) = n - |M \cap B_{\pi_1}|$. 

Therefore any perfect matching $M$ of $K_{n,n}$ satisfies conditions 1 and 2 of~\Cref{obs:pm-inters} if and only if $w(M) = n - b_1$. 

Now, let $t := |\mathcal{C}| + |\mathcal{F}|$. 
 Let $C_1, \ldots, C_m$ be an arbitrary enumeration of $\mathcal{C}$, and $e_1, \ldots, e_f$ an arbitrary enumeration of $\mathcal{F}$.
For each $1 \leq i \leq t$, let $\mathbf{e}_i \in \mathbb{F}_2^t$ denote the $i$-th standard basis vector. We now define a label function $\gamma\colon E(K_{n,n})\to \mathbb{F}_2^t$ by
\begin{equation}\label{eq:gammadef}
\gamma(e) = \left(\sum_{i \text{ s.t. }\! e \in C_i} \mathbf{e}_i\right) +  \left(\sum_{i=1}^f \mathbf{e}_{i+m} \cdot \mathbbm{1}_{[e  =e_i]}\right).
\end{equation}


 We are now ready for the main claim that will prove the theorem. Let us write $\mathbf{1}$ for the vector of all ones in $\mathbb{F}_2^t$.

\begin{claim}\label{claim:rp}
    There exists a secret $\sigma^\star$ compatible with the transcript $(\pi_1, \dots, \pi_T), (b_1,\dots, b_T)$ if and only if 
$K_{n,n}$ has a perfect matching $M$ such that 
\(
\gamma(M) := \sum_{e \in M} \gamma(e) = \mathbf 1
\)
and \(w(M) = n-b_1\).
\end{claim}
\begin{claimproof}
We have shown above that a matching $M$ of $K_{n, n}$ satisfies conditions 1 and 2 of~\Cref{obs:pm-inters} if and only if $w(M) = n - b_1$.  

Now we show that $M$ satisfies conditions 3 and 4 if and only if $\gamma(M) = \mathbf{1}$.
Let us use notation $\mathbf{v}[j]$ for the $j$-th coordinate of a vector $\mathbf{v}$, and we will show that $\gamma(M)[j]=1$ for $j \in \{1, \dots, m+f\}$. We consider two cases:

\begin{itemize}
    \item \textbf{(Case 1: $j \in \{m+1, \dots, m+f\}$). } In this case, $j=i+m$ for a unique $i\in\{1,\dots,f\}$, so only the second sum of~\eqref{eq:gammadef} is non-zero, so we have
\[
\gamma(M)[j] = \sum_{e \in M} \mathbbm{1}_{[e = e_i]} = |M \cap \{ e_i \}| \bmod  2 = |M \cap \{ e_i \}|,
\]
from where $\gamma(M)[j] = 1$ if and only if $|M \cap \{e_i\}| = 1$ (i.e., condition 3. holds for $e_i$). 
    \item \textbf{(Case 2: $j \in \{1, \dots, m\}$). } In this case, only the first sum of~\eqref{eq:gammadef} is non-zero, so we have
\[
    \gamma(M)[j] = \sum_{e \in M}\; \sum_{i \text{ s.t. }\!e \in C_i} \mathbbm{1}_{i = j} = \sum_{e \in M} \mathbbm{1}_{e \in C_j} = |M \cap C_j| \bmod 2,
\] from where $\gamma(M)[j] = 1$ if and only if $|M \cap C_j|$ is odd. But since $|C_j| = 2$, the only possibility is $|M \cap C_j| = 1$. Thus, $\gamma(M)[j] = 1$ if and only if condition 4 holds for $C_j$. 
\end{itemize}
This concludes the proof.
\end{claimproof}

Applying~\Cref{thm:dim-matching} to the instance $(G,w,\gamma,\mathbf 1)$, and then using~\Cref{claim:rp}, we conclude the proof. The final runtime results from the fact that $\max_{e} w(e) = n$.
\end{proof}